\documentclass[aps,showpacs,pra,twocolumn]{revtex4-1}

\usepackage{exscale}
\usepackage{graphicx}
\usepackage{amsmath}
\usepackage{latexsym}
\usepackage[caption=false]{subfig}
\usepackage{amsfonts}
\usepackage{amssymb}
\usepackage{bbm}
\usepackage{times}
\usepackage[T1]{fontenc}
\usepackage{lipsum}
\usepackage{amsthm}
\usepackage{fancyhdr,txfonts,bbm}

\usepackage{graphics}
\usepackage{bbold}
\usepackage{bm}
\usepackage{subfig}
\usepackage{color}

\usepackage{epsfig,pstricks}
\usepackage{changes}
\definechangesauthor[name={Markus}, color=red]{M}
\definechangesauthor[name={Dario}, color=magenta]{D}

\newtheorem{theorem}{Theorem}

\newtheorem{corollary}{Corollary}

\newtheorem*{ptheorem}{Pauli's Fundamental Theorem}

\newcommand{{\Cd}}{{\mathbb{C}^d}}
\newcommand{{\C}}{{\mathbb{C}}}

\DeclareMathOperator{\Tr}{Tr}

\begin{document}

\title{Candidate entanglement invariants for two Dirac spinors}

\author{Markus Johansson}
\affiliation{ICFO-Institut de Ciencies Fotoniques, The Barcelona Institute of Science and Technology,
08860 Castelldefels (Barcelona), Spain}

\date{\today}
\begin{abstract}
We consider two spacelike separated Dirac particles and construct five invariants under the spinor representations of the local proper orthochronous Lorentz groups. All of the constructed Lorentz invariants are identically zero for product states.
The behaviour of the Lorentz invariants under local unitary evolutions that act unitarily on any subspace with fixed particle momenta is studied.
All of the Lorentz invariants have invariant absolute values on such subspaces if the evolutions are generated by local zero-mass Dirac Hamiltonians. Some of them also for the case of nonzero-mass.
Therefore, they are considered potential candidates for describing spinor entanglement of two Dirac particles, with either zero or arbitrary mass.
Furthermore, their relations to the Wootters concurrence is investigated and their representations in the Foldy-Wouthuysen picture is given.

\end{abstract}
\maketitle
\section{Introduction}
The Dirac equation \cite{dirac2,dirac} is used in the Standard Model to describe leptons and quarks  \cite{schwartz}, in the Yukawa model of hadrons to describe baryons \cite{yukawa}, and in relativistic quantum chemistry \cite{pykk}. Without a mass term the Dirac equation admits solutions with definite chirality, so called Weyl particles \cite{weyl}. Dirac-like equations are also used to describe Dirac and Weyl particles in solid state and molecular systems as well as photonic crystals \cite{semenoff,bohm,lu,liu,pirie}. Emergent massless Dirac fermions have been experimentally demonstrated in graphene \cite{novos} and massless charged Weyl fermions have been demonstrated in Weyl semimetal \cite{Xu}.

A quantum system is entangled if it is in a superposition where some property of one subsystem is conditioned on a property of another subsystem. Then this property of the subsystem cannot be described independently of the property of the other subsystem.
The state of two spacelike separated entangled subsystems can therefore not be fully described in terms of local variables but have properties that require a non-local description \cite{epr,bell,chsh}. A commonly studied type of entanglement in non-relativistic quantum mechanics is that between the spins of two spacelike separated spin-$\frac{1}{2}$ particles \cite{bell,bennett,wootters,wootters2}. It can be characterized by the Wootters concurrence \cite{wootters,wootters2}, a polynomial in the state coefficients that is invariant up to a U(1) phase under local unitary evolution of the spins. In relativistic quantum mechanics the spinorial degree of freedom of a spin-$\frac{1}{2}$ particle is described by a four component Dirac spinor. 
Entanglement and the associated non-locality between the spinors of Dirac particles, as well as other descriptions of entanglement between Dirac particles, has previously been investigated and discussed in multiple works, see e.g. \cite{czachor,alsing,pachos,mano,caban3,caban,geng,leon,delgado,moradi,caban2,tessier,terno2,ahn,tera,tera2,terno,adami,won}.  

In this work we investigate the description of entanglement between the spinorial degrees of freedom of two Dirac particles from an algebraic point of view.
In particular we consider the problem of how to construct polynomial entanglement invariants \cite{hilbert,grassl,mumford} on a subspace defined by fixed particle momenta and spanned by spinorial degrees of freedom. 
We therefore search for polynomials that are zero for states prepared using only local resources, but not for all states. Moreover, we require that the polynomials are invariant up to a U(1) phase under local unitary evolution on any subspace defined by fixed particle momenta. Finally, we require that they are invariant under action by the spinor representation of the local proper orthochronous Lorentz groups. Such polynomials are considered potential candidates for describing entanglement of the spinorial degrees of freedom.

To construct such candidate entanglement invariants, we make the assumption that the state of any two spacelike separated Dirac spinors can be expanded in a basis that is formed from tensor products of the local basis elements used to describe single spinors. 
Given this assumption we derive five Lorentz invariants and describe their properties under local unitary evolutions generated by Dirac-like Hamiltonians with as well as without a mass term. We discuss their relation to the Wootters concurrence \cite{wootters,wootters2}, and finally investigate their representation in the Foldy-Wouthuysen picture \cite{foldy}.

The outline of this paper is as follows. Sections \ref{dir}-\ref{ham} review the relevant background material, discuss the physical assumptions made and introduce the tools used to construct the Lorentz invariants. 
In particular Section \ref{dir} introduces the description of Dirac and Weyl particles and discusses the physical assumptions. Section \ref{rep} describes the spinor representation of the Lorentz group and the charge conjugation transformation. In section \ref{invariants} we describe how to construct bilinear forms invariant under the spinor representation of the local proper orthochronous Lorentz transformations. 
Section \ref{ham} describes the behaviour of these bilinear forms under local unitary evolution generated by Dirac-like Hamiltonians.
Sections \ref{ent}-\ref{wout} contain the results.
In particular Section \ref{ent} describes how to construct candidate entanglement invariants and gives five such invariants which is the main results of this work. In Section \ref{exx} we consider a few examples of spinor entangled states. Section \ref{wout} gives the representation of the Lorentz invariants in the Foldy-Wouthuysen picture. Section \ref{diss} is the discussion and conclusions.
\section{Dirac spinors}\label{dir}

The Dirac equation was originally introduced in Ref. \cite{dirac2} to describe a relativistic spin-$\frac{1}{2}$ particle, or Dirac particle. For a particle with mass $m$ and charge $q$ in an electromagnetic four-potential $A_{\mu}(x)$ it can be written, with natural units $\hbar=c=1$, as
\begin{eqnarray}
\left[\sum_{\mu}\gamma^\mu(i\partial_{\mu}-qA_{\mu}(x)) -m\right]\psi(x)=0,
\end{eqnarray}
or equivalently as
\begin{eqnarray}\label{ggh}
\left[-\sum_{\mu=1,2,3}\gamma^0\gamma^\mu(i\partial_{\mu}-qA_{\mu}(x))+qA_0(x)I +m\gamma^0\right]\psi(x)=i\partial_0\psi(x),\nonumber\\
\end{eqnarray}
where $\psi(x)$ is a four component spinor 

\begin{eqnarray}\label{spinor}
\psi(x)\equiv
\begin{pmatrix}
\psi_0(x) \\
\psi_1(x)\\
\psi_2(x) \\
\psi_3(x) \\
\end{pmatrix},
\end{eqnarray}
which is a function of the four-vector $x$, and $\gamma^0,\gamma^1,\gamma^2,\gamma^3$ are $4\times 4$ matrices
that satisfy the Clifford relations
\begin{eqnarray}
\{\gamma^\mu,\gamma^\nu\}=2g^{\mu\nu}I,
\end{eqnarray}
where $g^{\mu\nu}$ is the Minkowski metric with signature $(+---)$.
The matrices $\gamma^0,\gamma^1,\gamma^2,\gamma^3$ are not uniquely defined by the Clifford relations and are chosen by convention.
One choice of these matrices that we use here is the so called Dirac matrices or { \it gamma matrices} defined as
\begin{align}
\gamma^0&=
\begin{pmatrix}
I & 0 \\
0 & -I  \\
\end{pmatrix},
&\gamma^1=
\begin{pmatrix}
0  & \sigma^1 \\
-\sigma^1 &  0 \\
\end{pmatrix},\nonumber\\
\gamma^2&=
\begin{pmatrix}
0  & \sigma^2 \\
-\sigma^2 &  0 \\
\end{pmatrix},
&\gamma^3=
\begin{pmatrix}
0  & \sigma^3 \\
-\sigma^3 &  0 \\
\end{pmatrix},
\end{align}
where $I$ is the $2\times 2$ identity matrix and $\sigma^1,\sigma^2,\sigma^3$ are the Pauli matrices
\begin{eqnarray}
I=
\begin{pmatrix}
1 & 0 \\
0 & 1  \\
\end{pmatrix},\phantom{o}
\sigma^1=
\begin{pmatrix}
0  & 1 \\
1 &  0 \\
\end{pmatrix},\phantom{o}
\sigma^2=
\begin{pmatrix}
0  & -i \\
i &  0 \\
\end{pmatrix},\phantom{o}
\sigma^3=
\begin{pmatrix}
1  & 0 \\
0 &  -1 \\
\end{pmatrix}.
\end{eqnarray}
For a derivation of the Dirac equation see e.g. Ref. \cite{dirac} Ch. XI. In the following we frequently suppress the four-vector dependence of $\psi(x)$ where it is not crucial and simply write $\psi$. 
The Dirac Hamiltonian $H_D$ can be identified in Eq. (\ref{ggh}) as 
\begin{eqnarray}
H_D=-\sum_{\mu=1,2,3}\gamma^0\gamma^\mu(i\partial_{\mu}-qA_{\mu})+qA_0I +m\gamma^0.
\end{eqnarray}

We can define two additional useful matrices in the algebra generated by the gamma matrices.
The first is the matrix
\begin{eqnarray}
C\equiv i\gamma^1\gamma^3=
\begin{pmatrix}
-\sigma^2 & 0 \\
0 & -\sigma^2  \\
\end{pmatrix},
\end{eqnarray}
which satisfies $C=C^\dagger=C^{-1}$ and for each gamma matrix $\gamma^\mu$ and its transpose $\gamma^{\mu T}$ it holds that
\begin{eqnarray}\label{uub}
\gamma^{\mu T}= C\gamma^\mu C.
\end{eqnarray}
The second is the matrix
\begin{eqnarray}
\gamma^5\equiv i\gamma^0\gamma^1\gamma^2\gamma^3=
\begin{pmatrix}
0 & I \\
I & 0  \\

\end{pmatrix},
\end{eqnarray}
which anticommutes with all the $\gamma^\mu$
\begin{eqnarray}
\gamma^5\gamma^\mu+\gamma^\mu\gamma^5=0,
\end{eqnarray}
and satisfies $\gamma^5=\gamma^{5\dagger}=(\gamma^{5})^{-1}$.

A modified Dirac equation without a mass term was considered by Weyl in Ref. \cite{weyl}
\begin{eqnarray}
\left[-\sum_{\mu=1,2,3}\gamma^0\gamma^\mu(i\partial_{\mu}-qA_{\mu})+qA_0I\right]\psi=i\partial_0\psi.
\end{eqnarray}
For this equation there are two invariant subspaces, defined by the projectors  $P_L=\frac{1}{2}(I-\gamma^5)$ and $P_R=\frac{1}{2}(I+\gamma^5)$, called the left- and right-handed chiral subspace, respectively. Solutions belonging to the right- or left-handed subspace, i.e., right-handed particles $\psi_R$ or left-handed particles $\psi_L$, are called Weyl particles and have the form 

\begin{eqnarray}\label{wey}
\psi_R=
\begin{pmatrix}
{\psi}_0 \\
{\psi}_1\\
{\psi}_0 \\
{\psi}_1 \\
\end{pmatrix},\phantom{o}
\psi_L=
\begin{pmatrix}
{\psi}_0 \\
{\psi}_1\\
-{\psi}_0 \\
-{\psi}_1 \\
\end{pmatrix}.
\end{eqnarray}

For any $t$ a solution to the Dirac equation can be expanded in basis modes $\phi_je^{i\bold{k}\cdot\bold{x}}$ as
\begin{eqnarray}
\psi(t,\bold{x})=\int_{\bold{k}}d\bold{k}\sum_{j}\psi_{j,\bold{k}}(t)\phi_je^{i\bold{k}\cdot\bold{x}},
\end{eqnarray}
where $\bold{k}$ is a wave three-vector, $\bold{x}$ is a spatial three-vector, $\psi_{j,\bold{k}}(t)$ are complex numbers, and $\phi_j$ is a basis for the spinor degree of freedom

\begin{eqnarray}\label{basis}
{\phi_0}=
\begin{pmatrix}
1 \\
0   \\
0   \\
0   \\
\end{pmatrix},\phantom{o}
{\phi_1}=
\begin{pmatrix}
0 \\
1   \\
0   \\
0   \\
\end{pmatrix},\phantom{o}
{\phi_2}=
\begin{pmatrix}
0 \\
0   \\
1   \\
0   \\
\end{pmatrix},\phantom{o}
{\phi_3}=
\begin{pmatrix}
0 \\
0   \\
0   \\
1   \\
\end{pmatrix}.
\end{eqnarray}

The Dirac inner product is defined as
\begin{eqnarray}\label{inn}
(\psi(t),\varphi(t))=\int_{\bold{x}}d\bold{x}\psi^\dagger(t,\bold{x})\varphi(t,\bold{x}).
\end{eqnarray}
We would like the basis modes to be orthogonal and normalizable with respect to this inner product. However, if we allow the modes to extend over all three-space, i.e., if we let them be infinite plane waves, the inner product $(\phi_je^{i\bold{k}\cdot\bold{x}},\phi_je^{i\bold{k'}\cdot\bold{x}})$ is ill defined. Since infinite plane waves are not supported on a bounded domain the inner product $(\phi_je^{i\bold{k}\cdot\bold{x}},\phi_je^{i\bold{k'}\cdot\bold{x}})$ does not converge for any $\bold{k}$ and $\bold{k'}$ and in particular is unbounded when $\bold{k}=\bold{k'}$. Thus the basis modes are neither orthogonal nor normalizable with respect to this inner product. Due to this problem Dirac opted to ad hoc impose the desired orthogonality relations $(\phi_je^{i\bold{k}\cdot\bold{x}},\phi_le^{i\bold{k'}\cdot\bold{x}})=\delta(\bold{k}-\bold{k'})\delta_{jl}$ where $\delta(\bold{k}-\bold{k'})$ is the Dirac delta and $\delta_{jl}$ is the Kronecker delta (See Ref. \cite{dirac} Ch. IV {\S} 23). Efforts to find a mathematically stringent framework where these relations hold led to the theory of generalized eigenfunctions and rigged Hilbert spaces \cite{gelfand,maurin}.
In this approach the momentum eigenmodes are not physically allowed states. Only Schwartz functions are allowed. These are functions that are smooth and such that for sufficiently large $\bold{x}$  the absolute values of the function and its derivatives to all orders decrease faster than any reciprocal power of $|\bold{x}|$ (See e.g. Ref. \cite{stein}). In this sense a Schwartz function is thus "localized" in a spatial region.
Moreover, Schwartz functions in spatial three-space are Schwartz functions also in momentum three-space. If the support in momentum space of a Schwartz function is compact and sufficiently small it becomes experimentally indistinguishable from the point support of a momentum eigenmode (See Appendix \ref{opp} for a discussion).

A different approach to have orthogonal and normalizable basis modes is to consider only modes in a finite spatial rectangular volume with periodic boundary conditions, so called box quantization (See e.g. Refs. \cite{mandl} and \cite{wightman}). Then the inner product as defined in Eq. (\ref{inn}) is well defined and there is a countable number of allowed momentum modes which satisfy the orthogonality relations $(\phi_je^{i\bold{k}\cdot\bold{x}},\phi_le^{i\bold{k'}\cdot\bold{x}})=V\delta_{\bold{k},\bold{k'}}\delta_{jl}$ where $\delta_{\bold{k},\bold{k'}}$ and $\delta_{jl}$ are Kronecker deltas and $V$ the volume of the box. For a sufficiently large box the discrete set of $\bold{k}$ becomes experimentally indistinguishable from a continuous set (See Appendix \ref{opp} for a discussion).
Note that box quantization is closely related to introducing an Infrared Cutoff (See e.g. Refs.   \cite{wightman} and \cite{duncan}).

In this work we use the momentum eigenbasis and assume orthogonality and normalizability of the modes. Moreover, as in References \cite{czachor,alsing,pachos,mano,moradi,caban,caban2,caban3} we consider states with definite momenta. This is strictly only possible if we assume boundary conditions that allow definite momenta in a finite spatial volume similarly to box quantization. In the rigged Hilbert space approach it can only be done as an approximation.
In the following we assume that the physical scenario is such that it is motivated to treat a particle as having a definite momentum while also being contained in a finite spatial volume, even if only as an approximate description.

With these caveats we can consider the four-dimensional subspace spanned by only spinorial degrees of freedom that is obtained by fixing $\bold{k}$, i.e., the subspace spanned by the modes $\phi_je^{i\bold{k}\cdot\bold{x}}$ for a fixed $\bold{k}$. The inner product $(\cdot,\cdot)_{\bold{k}}$ on such a subspace reduces to

\begin{eqnarray}
(\psi(t),\varphi(t))_{\bold{k}}=\psi^\dagger(t)\varphi(t).
\end{eqnarray}

Similarly to References \cite{alsing,pachos,moradi,caban2,caban3} we make the assumption that the state of any two spacelike separated particles that have not interacted can be described as a tensor product of single particle states $\psi_1(t)\otimes \psi_2(t)$. Furthermore, we assume that the tensor products of elements of the single particle bases $\phi_{j_1}e^{i\bold{k_1}\cdot\bold{x_1}}\otimes \phi_{j_2}e^{i\bold{k_2}\cdot\bold{x_2}}$ is a basis for the two-particle states.

In the same vein we describe the two particles as belonging to different Minkowski spaces. These two spaces should be understood as the local descriptions of spacetime used by two laboratories holding the respective particles. They could be the same Minkowski space described by two different spacelike separated observers or alternatively Minkowski tangent-spaces of two different points in spacetime in the sense of General Relativity (See e.g. Ref. \cite{wald}).

The assumption that a tensor product structure is appropriate to describe the state of the two particles in this situation is not trivial. The motivation for using it is the assumption that operations on one particle can be made jointly with operations on the other, i.e., that such operations commute. However, it is not clear that a description in terms of commuting operator algebras and a description in terms of tensor product spaces are always equivalent \cite{navascues,tsirelson,werner}. This open question is known as Tsirelson's Problem but for the case of two observers and finite dimensional algebras of operators it has been argued that a tensor product structure is not a restrictive assumption \cite{tsirelson,werner}.

The Dirac or Weyl particles in solid state or molecular systems are quasiparticles with a linear dispersion relation. The physical interpretation of such particles and the Hamiltonians that describe their evolution is thus quite different from that of Dirac or Weyl particles in relativistic quantum mechanics. Nevertheless, they can be described by four component spinors obeying Dirac-like equations.
An effective Dirac particle in the 2D Dirac semimetal graphene is described by a Hamiltonian that can be expressed on the form 

\begin{eqnarray}\label{2d}
H_{2D}=iv_D\gamma^0\sum_{\mu=1,2}\gamma^\mu\partial_\mu-\mu_PI,
\end{eqnarray}
where $v_D$ is the Dirac velocity and $\mu_P$ is the deviation of the chemical potential from its half-filling value (See e.g. Ref. \cite{kotov} or \cite{vass}). 
Similarly, the Hamiltonian for a massless linear Dirac particle in a 3D Dirac semimetal can be expressed on the general form \cite{bohm}
\begin{eqnarray}\label{3d}
H_{3D}=iv_D\gamma^0\sum_{\mu=1,2,3}\gamma^\mu\partial_\mu.
\end{eqnarray}
Examples of 3D Dirac semimetals include sodium bismuthide (Na$_3$Bi) \cite{fang,liu} and cadmium arsenide (Cd$_3$As$_2$) \cite{fang2,neupane,cava}.

\section{Spinor representation of the Lorentz group and the charge conjugation}\label{rep}

In a spacetime described by General Relativity one can at every non-singular point define the tangent vector space. This tangent space is isomorphic to the Minkowski space (See e.g. Ref. \cite{wald}).
Here we make the assumption that we can neglect the curvature of spacetime and describe the Dirac particle as belonging to the Minkowski tangent space rather than the underlying spacetime manifold.

A Lorentz transformation on the local Minkowski tangent space to a point in spacetime does not only transform the tangent space coordinates, but also induces an action on the spinor in the point. This action is given by a representation of the Lorentz transformation. Let $\Lambda$ be the Lorentz transformation and $S(\Lambda)$ the induced action on the spinor. Then $\psi(x)\to \psi'(x')=S(\Lambda)\psi(x)$ where $x'=\Lambda x$ (See e.g. Ref. \cite{zuber}), and the Dirac equation transforms as
\begin{eqnarray}
&&\left[\sum_{\mu}\gamma^\mu(i\partial_{\mu}-qA_{\mu}) -m\right]\psi(x)=0\nonumber\\
\to&&\left[\sum_{\mu,\nu}\gamma^\mu(\Lambda^{-1})^{\nu}_{\mu}(i\partial_{\nu}-qA_{\nu}) -m\right]S(\Lambda)\psi(x)=0.
\end{eqnarray}
Invariance of the equation implies that
\begin{eqnarray}
S^{-1}(\Lambda)\gamma^\mu S(\Lambda)=\sum_\nu\Lambda^{\mu}_{\nu}\gamma^\nu.
\end{eqnarray}
Thus, a Lorentz transformation 
corresponds in this way to a transformation of the matrices $\gamma^\mu$, but due to the invariance of the Dirac equation the new matrices $\gamma'^\mu$ must satisfy the same Clifford relations $\{\gamma'^\mu,\gamma'^\nu\}=2g^{\mu\nu}I$.
By Pauli's Fundamental Theorem this implies that the transformation $S(\Lambda)$ is non-singular and unique up to a constant factor.

\begin{ptheorem}
If for two sets of $4\times4$ matrices $\gamma^a$ and $\gamma'^a$ it holds that $\{\gamma^\mu,\gamma^\nu\}=2g^{\mu\nu}I=\{\gamma'^\mu,\gamma'^\nu\}$, then there exist a nonsingular $S$ such that $\gamma'^\mu=S\gamma^\mu S^{-1}$, and $S$ is unique up to a multiplicative constant.
\end{ptheorem}
\begin{proof}
See e.g. Ref. \cite{pauli} or \cite{messiah}.
\end{proof}

The Lorentz group is a Lie group and the connected component that contains the identity element, the proper orthochronous Lorentz group, can be generated by exponentials of its Lie algebra. Likewise, the representation of the proper orthochronous Lorentz group acting on a spinor is a connected Lie group.
The generators $S^{\rho\sigma}$ of the Lie algebra of the spinor representation are defined by

\begin{eqnarray}\label{gene}
S^{\rho\sigma}=\frac{1}{4}[\gamma^\rho,\gamma^\sigma]=\frac{1}{2}\gamma^\rho\gamma^\sigma-\frac{1}{2}g^{\rho\sigma}I.
\end{eqnarray}
A finite transformation can be obtained as an exponential
 
\begin{eqnarray}
S(\Lambda)=\exp\left(\frac{1}{2}\sum_{\rho,\sigma} \omega_{\rho\sigma}S^{\rho\sigma}\right),
\end{eqnarray}
where $\omega_{\rho\sigma}$ are real numbers. Products of these finite transformations can be used to describe 
the spinor representation of any proper orthochronous Lorentz transformation. 
See e.g. Ref. \cite{zuber}. This holds since the spinor representation of the proper orthochronous Lorentz group is a connected matrix Lie group (See e.g. Ref. \cite{hall} Ch. 3.8).

The Lorentz group has three other connected components besides the  proper orthochronous subgroup. These are related to the proper orthochronous subgroup by the parity inversion transformation P , the time reversal transformation T, and the combined parity time transformation PT, respectively. Likewise, the spinor representation of the Lorentz group has three other connected components
related to the spinor representation of the proper orthochronous subgroup by the the spinor representations of the parity inversion P, the time reversal T, and the parity time transformation PT, respectively. The spinor representations of P and T are defined up to multiplication by a U(1) phase that is chosen by convention.
The representation of the parity transformation on the spinor can be chosen as

\begin{eqnarray}
S(\textrm{P})=\gamma^0.
\end{eqnarray}
The representation of the time reversal transformation involves the matrix $C$ and a complex conjugation of the spinor. Time reversal T of a spinor $\psi$ can be chosen as $\psi \to C\psi^*$. 

In addition to the Lorentz group we may consider also the charge conjugation transformation C, charge parity CP, as well as the charge parity time CPT transformation. As for P and T, the spinor representation of the charge conjugation is defined up to multiplication by a U(1) phase that is chosen by convention.
Charge conjugation like time reversal involves complex conjugation and can be chosen as $\psi \to i\gamma^2\psi^*$. The CP transformation is the combination of charge conjugation and parity inversion and is given by $\psi \to -i\gamma^0\gamma^2\psi^*= iC\gamma^5\psi^*$. 
The CPT transformation is the combination of charge conjugation, parity inversion, and time reversal and is given by the matrix $-i\gamma^5$

\begin{eqnarray}
S(\textrm{CPT})=-i\gamma^5.
\end{eqnarray}
See e.g. Ref. \cite{bjorken}.
In the following we use these conventions for the spinor representations of the P, T, and C transformations.

\section{Invariants of the spinor representation of the proper orthochronous Lorentz group}
\label{invariants}
A physical quantity is said to be Lorentz covariant if it transforms under some representation of the Lorentz group.  In particular, a Lorentz covariant scalar that remains the same under all Lorentz transformations is said to be a Lorentz invariant scalar. A Lorentz covariant scalar that changes sign under parity inversion but is invariant under all other Lorentz transformations is said to be a Lorentz pseudo-scalar. In the following we refer to a covariant scalar that is invariant under the proper orthochronous Lorentz group as a Lorentz invariant. Thus, we do not distinguish between Lorentz scalars and pseudo-scalars and call both Lorentz invariants.

From the form of the generators $S^{\rho\sigma}$ of the spinor representation of the proper orthochronous Lorentz group in Eq. (\ref{gene}) and the properties of the matrix $C$ given in Eq. (\ref{uub}) we see that
\begin{eqnarray}
S^{\rho\sigma T}=\frac{1}{4}[\gamma^{\sigma T},\gamma^{\rho T}]=-\frac{1}{4}C[\gamma^{\rho},\gamma^{\sigma}]C=-CS^{\rho\sigma}C.
\end{eqnarray}
Therefore, for a finite transformation $S(\Lambda)$ we have that $S(\Lambda)^T=CS(\Lambda)^{-1}C$ or equivalently $S(\Lambda)^TC=CS(\Lambda)^{-1}$.
Using this we can construct a Lorentz invariant from spinors $\psi$ and $\varphi$ as a bilinear form
\begin{eqnarray}
\psi^TC\varphi,
\end{eqnarray}
which transforms as $\psi^TS(\Lambda)^TCS(\Lambda)\varphi=\psi^TCS(\Lambda)^{-1}S(\Lambda)\varphi=\psi^TC\varphi$ under the spinor representations of proper orthochronous Lorentz transformations (See e.g. Ref. \cite{pauli}). Moreover, since $\gamma^0=(\gamma^{0})^{T}=(\gamma^{0})^{-1}$ and $\gamma^0C=C\gamma^0$ it follows that $\psi^TS(\textrm{P})^TCS(\textrm{P})\varphi=\psi^T\gamma^0C\gamma^0\varphi=\psi^TC\varphi$. Thus, $\psi^TC\varphi$ is invariant under parity transformation.
Since $S(\textrm{CPT})=-i\gamma^{5}$ and $\gamma^5C=C\gamma^5$ and $\gamma^5=(\gamma^{5})^{T}=(\gamma^{5})^{-1}$ it follows that $\psi^TS(\textrm{CPT})^TCS(\textrm{CPT})\varphi=-\psi^T\gamma^5C\gamma^5\varphi=-\psi^TC\varphi$. Thus, $\psi^TC\varphi$ is not invariant under CPT transformations.

Next we can see that since $\gamma^5$ anti-commutes with all $\gamma^{\mu}$ it commutes with the generators $S^{\rho\sigma}$
\begin{eqnarray}
[S^{\rho\sigma},\gamma^5]=\frac{1}{4}[\gamma^{\rho},\gamma^{\sigma}]\gamma^5-\frac{1}{4}\gamma^5[\gamma^{\rho},\gamma^{\sigma}]=0,
\end{eqnarray}
and thus $\gamma^5$ commutes with the spinor representations of proper orthochronous Lorentz transformations $S(\Lambda)\gamma^5=\gamma^5S(\Lambda)$.
Therefore, we can construct an invariant of the spinor representation of the proper orthochronous Lorentz group as the bilinear form
\begin{eqnarray}
\psi^TC\gamma^5\varphi,
\end{eqnarray}
which transforms as $\psi^TS(\Lambda)^TC\gamma^5S(\Lambda)\varphi=\psi^TCS(\Lambda)^{-1}\gamma^5S(\Lambda)\varphi=\psi^TC\gamma^5S(\Lambda)^{-1}S(\Lambda)\varphi=\psi^TC\gamma^5\varphi$.
Moreover, since $\gamma^0\gamma^5=-\gamma^5\gamma^0$ it follows that $\psi^TS(\textrm{P})^TC\gamma^5S(\textrm{P})\varphi=\psi^T\gamma^0C\gamma^5\gamma^0\varphi=-\psi^TC\gamma^5\varphi$. Thus, $\psi^TC\gamma^5\varphi$ is not invariant under parity transformation but changes sign.
Since $S(\textrm{CPT})=-i\gamma^{5}$ it follows that $\psi^TS(\textrm{CPT})^TC\gamma^5S(\textrm{CPT})\varphi=-\psi^T\gamma^5C\gamma^5\gamma^5\varphi=-\psi^TC\gamma^5\varphi$. Thus, $\psi^TC\gamma^5\varphi$ is not invariant under CPT transformations but changes sign.

\section{The behaviour of the Lorentz invariants under unitary spinor evolution generated by Dirac-like Hamiltonians}\label{ham}
Here we consider a subspace spanned by only spinorial degrees of freedom, i.e., a subspace spanned by $\phi_je^{i\bold{k}\cdot\bold{x}}$ for a fixed $\bold{k}$. We consider an evolution that is generated by a Hamiltonian operator $H$ and acts unitarily on all such subspaces. Then it is required that $(\phi_je^{i\bold{k}\cdot\bold{x}},H\phi_le^{i\bold{k'}\cdot\bold{x}})\propto\delta_{\bold{k},\bold{k'}}$, otherwise the subspaces are not invariant under the evolution. Therefore, to have unitary action on the subspaces we consider Hamiltonians that do not depend on $\bold{x}$.

We consider again the inner product
\begin{eqnarray}
({\psi(t)},{\varphi(t)})_{\bold{k}}=\psi^{\dagger}(t)\varphi(t),
\end{eqnarray}
and assume that $H(s)$ is bounded and strongly continuous, i.e., $\lim_{t\to s}||H(t){\psi}-H(s){\psi}||=0$ for all ${\psi}$ and $s$ where $||\cdot||$ is the norm induced by the inner product. Then we have the following theorem.

\begin{theorem}\label{th}
Assume that $t\in\mathbb{R}\to H(t)$ is a strongly continuous map into the bounded Hermitian operators on a Hilbert space $\mathcal{H}$. Then there exists an evolution operator $U(t,s)$ such that for all ${\psi}\in\mathcal{H}$ it holds that
${\psi(t)}=U(t,s){\psi(s)}$ and $\partial_{t}U(t,s)=-iH(t)U(t,s)$. Moreover, the evolution operator satisfies $U(r,s)U(s,t)=U(r,t)$ and $U(t,t)=I$ and can be expressed as an ordered exponential
\begin{eqnarray*}
U(t,r)=&&\mathcal{T}_{\leftarrow}\{e^{-i\int_{r}^tH(s)ds}\}\nonumber\\
\equiv&&\sum_{n=0}^\infty(-i)^n\int_{r}^t\int_{r}^{s_n}\int_{r}^{s_{n-1}}\dots\nonumber\\
&&\dots\int_{r}^{s_{2}}H(s_n)\dots H(s_{1})ds_1\dots ds_{n-2} ds_{n-1} ds_{n}.\nonumber\\
\end{eqnarray*}
\end{theorem}
\begin{proof}
See e.g. Ref. \cite{reed}.
\end{proof} 
We can consider the conjugate transpose, and complex conjugate in the given basis, of the evolution operator
\begin{eqnarray}
U(t,r)&&=\mathcal{T}_{\leftarrow}\{e^{-i\int_{r}^tH(s)ds}\}\nonumber\\
U(t,r)^\dagger &&=\mathcal{T}_{\rightarrow}\{e^{i\int_{r}^tH(s)ds}\}\nonumber\\
U(t,r)^*&&=\mathcal{T}_{\leftarrow}\{e^{i\int_{r}^t H^T(s)ds}\}\nonumber\\
(U(t,r)^\dagger )^*&&=\mathcal{T}_{\rightarrow}\{e^{-i\int_{r}^t H^T(s)ds}\}.
\end{eqnarray}

Assume now that for a time independent $X$ it holds that $XH(s)=-H(s)^TX$ for all $s$. Then we can see that $X\mathcal{T}_{\leftarrow}\{e^{-i\int_{0}^tH(s)ds}\}=\mathcal{T}_{\leftarrow}\{e^{i\int_{0}^tH(s)^Tds}\}X$, and it follows that
\begin{eqnarray}
&&{\psi^T} (U(t,0)^\dagger )^* XU(t,0){\varphi}\nonumber\\&&={\psi^T}\mathcal{T}_{\rightarrow}\{e^{-i\int_{0}^t H^T(s)ds}\}
\mathcal{T}_{\leftarrow}\{e^{i\int_{0}^t H^T(s)ds}\}X{\varphi}={\psi^T}X{\varphi}.
\end{eqnarray}

We now investigate the candidates $X=C$ and $X=C\gamma^5$.
For products of distinct gamma matrices we have that

\begin{eqnarray}\label{c}
(\gamma^\mu)^TC&=&C\gamma^\mu\nonumber\\
(\gamma^\mu\gamma^\nu)^TC&=&-C\gamma^\mu\gamma^\nu\nonumber\\
(\gamma^\mu\gamma^\nu\gamma^\rho)^TC&=&-C\gamma^\mu\gamma^\nu\gamma^\rho\nonumber\\
(\gamma^\mu\gamma^\nu\gamma^\rho\gamma^\sigma)^TC&=&C\gamma^\mu\gamma^\nu\gamma^\rho\gamma^\sigma,
\end{eqnarray}
and
\begin{eqnarray}\label{c5}
(\gamma^\mu)^TC\gamma^5&=&-C\gamma^5\gamma^\mu\nonumber\\
(\gamma^\mu\gamma^\nu)^TC\gamma^5&=&-C\gamma^5\gamma^\mu\gamma^\nu\nonumber\\
(\gamma^\mu\gamma^\nu\gamma^\rho)^TC\gamma^5&=&C\gamma^5\gamma^\mu\gamma^\nu\gamma^\rho\nonumber\\
(\gamma^\mu\gamma^\nu\gamma^\rho\gamma^\sigma)^TC\gamma^5&=&C\gamma^5\gamma^\mu\gamma^\nu\gamma^\rho\gamma^\sigma.
\end{eqnarray}
From this we can conclude that a Hamiltonian with terms of the form
\begin{eqnarray}\label{w2}
H^{2,3}(t)=\gamma^\mu\gamma^\nu\phi_{\mu\nu}(t)+\gamma^\mu\gamma^\nu\gamma^\rho\kappa_{\mu\nu\rho}(t),
\end{eqnarray}
satisfies $CH^{2,3}(t)=-(H^{2,3}(t))^TC$. Moreover, a Hamiltonian with terms of the form
\begin{eqnarray}\label{w1}
H^{1,2}(t)=\gamma^\mu\eta_{\mu}(t)+\gamma^\mu\gamma^\nu\lambda_{\mu\nu}(t),
\end{eqnarray}
satisfies $C\gamma^5H^{1,2}(t)=-(H^{1,2}(t))^TC\gamma^5$.
A Hamiltonian term $H^0(t)=f(t)I$ proportional to the identity clearly is its own transpose and commutes with both $C$ and $C\gamma^5$.

The Dirac Hamiltonian has a term that is first degree in gamma matrices, the mass term $m\gamma^0$, and a second degree term $\sum_{\mu=1,2,3}\gamma^{0}\gamma^{\mu}(i\partial_{\mu}-qA_{\mu}(t))$.
Moreover, it has a term proportional to the identity, the coupling to the scalar potential $qA_0(t)I$. However, we can perform a change of variables to remove any such zeroth degree term from the Hamiltonian. If we define
$\psi'= e^{-i\theta(t)}\psi$ the new Hamiltonian $H'$ satisfying $i\partial_t\psi'(t)=H'\psi'(t)$ is 
$H'=H+\gamma^0\sum_\mu\gamma^\mu\partial_\mu \theta(t)$. This amounts to a change of local U(1) gauge (See e.g. Ref. \cite{griffiths}). By choosing $\theta(t)=-q\int_{t_0}^tA_0(s)ds$ we see that the term proportional to the identity in $H'$ is $qA_0(t)I-q\partial_t \int_{t_0}^tA_0(s)dsI=0$. Apart from the zeroth degree term $H'$ in general has terms of the same degrees in the gamma matrices as $H$ and cannot gain terms with degrees different from those of $H$. Let $U'(t,t_0)$ be the evolution generated by $H'$. Then $\psi(t)=e^{-iq\int_{t_0}^tA_0(s)ds}\psi'(t)=e^{-iq\int_{t_0}^tA_0(s)ds}U'(t,t_0)\psi'(t_0)=e^{-iq\int_{t_0}^tA_0(s)ds}U'(t,t_0)\psi(t_0)$.

Therefore, for any evolution $U_D(t,0)$ generated by Dirac-type Hamiltonians  we can see that 

\begin{eqnarray}
&&\psi^{T}(U_D(t,0)^\dagger)^*C\gamma^5U_D(t,0)\varphi\nonumber\\
&&=e^{-2iq\int_{0}^tA_0(s)ds}{\psi^T}({U'}_D(t,0)^\dagger)^* C\gamma^5 U'_D(t,0){\varphi}\nonumber\\
&&=e^{-2iq\int_{0}^tA_0(s)ds}{\psi^T}C\gamma^5{\varphi}.
\end{eqnarray}

Likewise, for any evolution $U_W(t,0)$ generated by zero-mass Dirac Hamiltonians we have 

\begin{eqnarray}
&&\psi^{T}(U_W(t,0)^\dagger)^*CU_W(t,0)\varphi \nonumber\\
&&=e^{-2iq\int_{0}^tA_0(s)ds}{\psi^T}({U'}_W(t,0)^\dagger)^*CU'_W(t,0){\varphi}\nonumber\\
&&=e^{-2iq\int_{0}^tA_0(s)ds}{\psi^T}C{\varphi}.
\end{eqnarray}

We can consider a few additional or alternative Hamiltonian terms from the literature. For example a coupling to a Yukawa scalar boson $g\gamma^0\phi$ \cite{yukawa}, but this term behaves analogously to the mass term. Coupling to a Yukawa pseudo-scalar boson $gi\gamma^0\gamma^5\phi$ on the other hand is cubic in the gamma matrices. So is a Pauli-coupling $i\gamma^0\gamma^\mu\gamma^\nu\partial_{\nu}A_\mu$ to the vector potential (See e.g. \cite{das}). An electroweak type chiral coupling term to a vector boson, e.g. $\sum_\mu g\gamma^0\gamma^\mu(I\pm\gamma^5)Z_{\mu}$ \cite{weinberg}, contains terms of degree 2 and 4 in the gamma matrices.

Thus,
${\psi^T}C\gamma^5{\varphi}$
is invariant under evolution generated by Dirac Hamiltonians up to a U(1) phase. More generally it is invariant under evolution generated by Hamiltonians of the type $H^{1,2}(t)+H^0(t)$ up to a U(1) phase.
Likewise,
${\psi^T}C{\varphi}$
is invariant up to a U(1) phase under evolution generated by zero-mass Dirac Hamiltonians. More generally it is invariant, up to a U(1) phase, under evolution generated by Hamiltonians of the type $H^{2,3}(t)+H^0(t)$, which could include a pseudo-scalar Yukawa term or a Pauli coupling.

Neither ${\psi^T}C{\varphi}$ nor ${\psi^T}C\gamma^5{\varphi}$ is invariant under evolution generated by Hamiltonians that contain both a mass (or Yukawa scalar) term and a Yukawa pseudo-scalar or Pauli coupling term. Moreover, neither is invariant under evolution generated by
Hamiltonians with electroweak type chiral coupling to a vector boson.

We can consider the bilinear forms ${\psi^T}C{\varphi}$ or ${\psi^T}C\gamma^5{\varphi}$ also in the context of Dirac or Weyl particles in solid state or molecular systems.
The Hamiltonian in Eq. (\ref{2d}) for Dirac particles in 2D Dirac semimetals  \cite{kotov,vass}, and 
the Hamiltonian in Eq. (\ref{3d}) for zero-mass Dirac particles in 3D Dirac semimetals \cite{bohm} have only second and zeroth degree terms in the gamma matrices.
Thus, in both cases ${\psi^T}C{\varphi}$ and ${\psi^T}C\gamma^5{\varphi}$ are invariant up to a U(1) phase.

Additional terms that can be introduced in the Hamiltonians is the Semenoff mass term $M_S\gamma^0\gamma^3$ \cite{semenoff} and the Haldane mass term $M_{H}\gamma^5\gamma^0\gamma^3$ \cite{haldane}. Since the Semenoff and Haldane mass terms are both quadratic in gamma matrices, ${\psi^T}C{\varphi}$ and ${\psi^T}C\gamma^5{\varphi}$ are still invariant up to a U(1) phase with these additions.

\section{Candidate entanglement invariants}\label{ent}

In this section we consider the issue of describing entanglement between two Dirac spinors and construct five candidate entanglement invariants which is the main result of this work. 

\subsection{Defining and describing spinor entanglement properties}\label{spen}

A system of two particles can in general be entangled in a multitude of qualitatively different ways.  
If two entangled states can be deterministically transformed into each other by local unitary evolution of the system and changes of local reference frames they may be considered to be entangled in equivalent ways. Moreover, if two states are identical up to multiplication by a constant factor they may be considered physically equivalent. Therefore, we may consider two entangled states to be equivalently entangled if the states can be transformed into each other by local unitary evolutions, changes of local reference frames and multiplication by a constant factor, and inequivalently entangled otherwise.
In the following we refer to the local unitary evolutions and changes of local reference frames collectively as {\it local reversible operations}.
In general a number of different properties of the entanglement can be identified and used to distinguish inequivalent types of entanglement. Any such property describing the entanglement must be unchanged by local reversible operations \cite{popescu,pop2,ben}. Moreover, no entanglement property should be present in states that can be created using only local resources, i.e., product states.
A general approach to the characterization of this kind of entanglement properties has been described in References \cite{popescu,pop2}.
We follow this general approach here and outline it below. 

Any state $\psi_{AB}$ of a system belongs to a set $\mathcal{O}_{\psi_{AB}}$ that consists of all states that can be obtained from $\psi_{AB}$ by local reversible operations. We refer to such a set $\mathcal{O}_{\psi_{AB}}$, as an {\it orbit} of the local reversible operations. Any two different orbits are disjoint and the Hilbert space can be fully decomposed into the collection of all such orbits. If the orbit $\mathcal{O}_{\psi_{AB}}$ can be isomorphically mapped to the orbit $\mathcal{O}_{\phi_{AB}}$ by a map $m_c:\psi_{AB}\to c\psi_{AB}$ for some $c\in \mathbb{C}-\{0\}$, i.e., if the two orbits are identical up to elementwise multiplication by a nonzero constant $c$, we can consider them physically equivalent. Let $\tilde{\mathcal{O}}_{\psi_{AB}}$ be the equivalence class of orbits that can be obtained from $\mathcal{O}_{\psi_{AB}}$ by maps $m_{c}$ for all $c\in \mathbb{C}-\{0\}$.

If two states belong to different equivalence classes they differ in some physical property that cannot be changed by local reversible operations.
Thus two entangled states $\psi_{AB}$ and $\phi_{AB}$ such that $\tilde{\mathcal{O}}_{\psi_{AB}}\neq\tilde{\mathcal{O}}_{\phi_{AB}}$ are by definition inequivalently entangled, i.e., entangled in qualitatively different ways. This approach for describing different kinds of entanglement in terms of inequivalence under local reversible operations has been used for various systems of non-relativistic spin-$\frac{1}{2}$ particles (See e.g. References \cite{ekert,grassl,popescu,pop2,kempe,lind2,car2,higuchi,tarrach,toni,sud,toumazet}).

A way to characterize the different inequivalent types of entanglement in a system is to find parameters that can distinguish between the different equivalence classes of entangled states. If we require that the parameters only distinguish different inequivalent forms of entanglement and not between any other properties, we must require that the parameters are functions that do not vary within any equivalence class. This implies that the parameters must be invariant under local reversible operations and invariant under multiplication of a state by any constant $c\in \mathbb{C}-\{0\}$.
We can construct such parameters by finding a set of functions $f_i$ on the Hilbert space that are invariant under local reversible operations with determinant 1 and homogeneous under multiplication of a state by $c\in \mathbb{C}-\{0\}$, i.e., functions $f_i$ invariant under local reversible operations with unit determinant and satisfying $f_i(c\psi_{AB})=c^{k(i)}f_i(\psi_{AB})$ where $k(i)$ is the degree of homogeneity of $f_i$. Then whenever $f_i\neq0$ and $f_j\neq0$ the ratio $f_i^{k(j)}/f_j^{k(i)}$ has degree of homogeneity zero. Such a ratio is thus both invariant under all local reversible operations and invariant under multiplication of a state by a constant $c\in \mathbb{C}-\{0\}$. Therefore this kind of functions provide coordinates parametrizing the set of equivalence classes $\tilde{\mathcal{O}}_{\psi_{AB}}$.
If we further require that the homogeneous functions $f_i$ are identically zero for all product states they can be used on their own as witnesses of entanglement. Any nonzero value of such a function implies that the state is entangled.
We refer to the homogeneous functions $f_i$ that are invariant under local reversible operations with determinant 1 and identically zero for all product states as {\it entanglement invariants}.

If for two entangled states there exists a ratio between two entanglement invariants with degree of homogeneity zero that takes different values for the two states these two states do not belong to the same equivalence class and are thus inequivalently entangled. However, if for two states all such ratios between entanglement invariants takes the same values, it is not necessarily the case that the two states belong to the same equivalence class.
Thus a given set of entanglement invariants may not be able to distinguish all inequivalent types of entanglement. If this is the case we say that the set only provides a {\it partial characterization} of the entanglement properties of the system.
In general it is not possible to find a set of entanglement invariants that distinguish all equivalence classes and the characterization provided by the invariants is only partial.

One way to construct entanglement invariants is as homogeneous polynomials in the state coefficients that are invariant under local reversible operations with determinant 1 and identically zero for all product states.
Characterizing entanglement using such polynomial invariants has previously been done for different systems of non-relativistic spin-$\frac{1}{2}$ particles (See e.g. References \cite{grassl,wootters,wootters2,kempe,sud,coffman,luque,verstraete2,miake,luq4,verstraete, toumazet}).

If we consider entanglement of Dirac spinors we may tentatively identify the local reversible operations acting on the spinors as the set of local unitary spinor evolutions generated by the set of allowed Dirac Hamiltonians together with the 
set of possible changes of local reference frames, i.e., the local spinor representations of the proper orthochronous Lorentz transformations. Then, for a system of two Dirac particles with definite momenta we have three conditions that tentatively define a spinor entanglement property for pure states of such particles.

\begin{enumerate}
  \item  [(1)] Non-existence for any state that can be created using only local resources, i.e., any product state.
  \item [(2)] Invariance under local evolutions generated by physically allowed Dirac Hamiltonians that act unitarily on any subspace defined by fixed momenta.
  \item[(3)]Invariance under changes of local inertial reference frames, i.e., Lorentz invariance.
\end{enumerate}

An alternative approach to describing entanglement focuses on entanglement properties that can be quantified. In this approach   a quantifiable entanglement property is required to be non-increasing on average under any local operations assisted by classical communication \cite{vidal}. This condition is called entanglement monotonicity \cite{vidal}. To evaluate this condition one needs to characterize the set of such local operations assisted by classical communication and introduce measures of entanglement \cite{entmes}, which goes beyond the scope of this work.

\subsection{Construction of invariants}\label{un}

Any entanglement invariant that can be used to describe spinor entanglement properties as defined in Section \ref{spen} of a system of two Dirac particles with definite momenta must take the value zero for any product state, but not for all states. Further it must be invariant under local evolutions with determinant 1 generated by Dirac Hamiltonians that act unitarily on any subspace defined by fixed momenta, i.e., any subspace spanned by spinorial degrees of freedom, and be Lorentz invariant.

We can construct quantities with these properties by utilizing the features of the bilinear forms described in Sect. \ref{invariants} ans Sect \ref{ham}.
The bilinear forms $\psi^T(x)C\gamma^5\varphi(x)$ and $\psi^T(x)C\varphi(x)$ are pointwise Lorentz invariant. For $\psi$ and $\varphi$ belonging to a subspace spanned by spinorial degrees of freedom they are invariant, up to a U(1) phase, under evolution that acts unitarily on the subspace and is generated by Dirac Hamiltonians and zero-mass Dirac Hamiltonians, respectively. Moreover, $\psi^T(x)C\psi(x)$ and $\psi^T(x)C\gamma^5\psi(x)$ are identically zero due to the antisymmetry of $C$ and $C\gamma^5$.

Now, consider two spacelike separated observers Alice and Bob each with their own laboratory containing a Dirac or Weyl particle. Let the two particles be in a joint state and assume that Alice's operations on the shared system can be made jointly with Bob's, i.e., assume that Alice's operations commute with Bob's. We assume that we can use a tensor product structure to describe the shared system and use the tensor products $\phi_{j_A}e^{i\bold{k_A}\cdot\bold{x_A}}\otimes \phi_{j_B}e^{i\bold{k_B}\cdot\bold{x_B}}$ of local basis elements as a basis. 
Then we can expand the state in this basis as
\begin{eqnarray}\label{bilbo}
\psi_{AB}(t)=\sum_{\bold{k_A},\bold{k_B}}\sum_{j_A,j_B}\psi_{j_A,j_B,\bold{k_A},\bold{k_B}}(t)\phi_{j_A}e^{i\bold{k_A}\cdot\bold{x_A}}\otimes \phi_{j_B}e^{i\bold{k_B}\cdot\bold{x_B}},
\end{eqnarray}
where $\psi_{j_A,j_B,\bold{k_A},\bold{k_B}}(t)$ are complex numbers. 

Next we assume that the state belongs to a subspace where $\bold{k_A}$ and $\bold{k_B}$ are fixed, i.e., a subspace spanned by the spinorial degrees of freedom.
Note that without this assumption local unitary evolution generated by Dirac Hamiltonians in general transforms a state where the spinorial degrees of freedom are entangled and not conditioned on the momenta, to a state where the spinorial degrees of freedom are conditioned on the momenta, even if it acts unitarily on every fixed momenta subspace. Then the spinor entanglement cannot be described independently of the momenta (See Appendix \ref{dwalin} for a discussion).
We suppress the indices $\bold{k_A},\bold{k_B}$ in the description of the state and let $\psi_{jk}\equiv \psi_{j_A,k_B,\bold{k_A},\bold{k_B}}$.
The coefficients $\psi_{jk}$ can be arranged as a matrix by letting $j$ be the row index and $k$ be the column index. Let us denote this matrix $\Psi_{AB}$. Written out it is 
\begin{eqnarray}\label{bofur}
\Psi_{AB}\equiv
\begin{pmatrix}
\psi_{00} & \psi_{01} & \psi_{02} & \psi_{03}\\
\psi_{10} & \psi_{11} & \psi_{12} & \psi_{13}\\
\psi_{20} & \psi_{21} & \psi_{22} & \psi_{23} \\
\psi_{30} & \psi_{31} & \psi_{32} & \psi_{33}\\
\end{pmatrix}.
\end{eqnarray}
Transformations $S_A$ on Alice's part of the system act from the left and transformations $S_B$ on Bob's part of the system  act in transposed form $S_{B}^T$ from the right

\begin{eqnarray}
\Psi_{AB}\to S_A\Psi_{AB}S_{B}^T.
\end{eqnarray}

Applying what we learned in Section \ref{invariants} we can now construct invariants under action of the spinor representation of the proper orthochronous Lorentz group in both Alice's lab and in Bob's lab.
A first simple invariant of degree 2 in the coefficients of the state is
\begin{eqnarray}
I_1=\frac{1}{2}\Tr[\Psi_{AB}^{T}C\Psi_{AB}C],
\end{eqnarray}
where the first $C$ is understood to be defined in the basis of Alice and the second in the basis of Bob.
The invariance of $I_1$ can be seen directly from its transformation properties under local spinor representations of proper orthochronous Lorentz transformations in Alice's lab $S_A(\Lambda_A)$ and Bob's lab $S_B(\Lambda_B)$

\begin{eqnarray}
&&\frac{1}{2} \Tr[\Psi_{AB}^{T}S_{A}^{T}(\Lambda_A)CS_A(\Lambda_A)\Psi_{AB}S_{B}^{T}(\Lambda_B)CS_B(\Lambda_B)]\nonumber\\
&=&\frac{1}{2} \Tr[\Psi_{AB}^{T}CS_{A}^{-1}(\Lambda_A)S_A(\Lambda_A)\Psi_{AB}CS_{B}^{-1} (\Lambda_B)S_B(\Lambda_B)]\nonumber\\&=&\frac{1}{2} \Tr[\Psi_{AB}^{T}C\Psi_{AB}C].
\end{eqnarray}
Moreover, $I_1$ is identically zero for product states. Written out in terms of coefficients the invariant is

\begin{eqnarray}
I_1=\psi_{00} \psi_{11}-\psi_{01} \psi_{10} +\psi_{02}\psi_{13} - \psi_{03} \psi_{12} \nonumber\\+\psi_{20}      \psi_{ 31} -\psi_{21}\psi_{30}+\psi_{22 }\psi_{33}  - 
 \psi_{23}\psi_{32}.
\end{eqnarray}
If we divide $\Psi_{AB}$ into four $2\times 2$ block matrices, two diagonal and two off-diagonal, we see that the invariant is the sum of the determinants of these blocks.

By construction $I_1$ is invariant under parity inversion in both Alice's and Bob's lab in addition to the proper orthochronous Lorentz group. Moreover, due to its construction it is invariant under unitary evolution generated by local zero-mass Dirac Hamiltonians, up to a U(1) phase, but not under evolution generated by nonzero-mass Dirac Hamiltonians.

We can construct three more invariants under the spinor representation of the proper orthochronous Lorentz group in a similar way. The Lorentz invariant

\begin{eqnarray}
I_{2}=\frac{1}{2} \Tr[\Psi_{AB}^{T}C\gamma^5\Psi_{AB}C\gamma^5],
\end{eqnarray}
is not invariant under P in either Alice's or Bob's lab. Furthermore, it is invariant under unitary evolution generated by local arbitrary mass Dirac Hamiltonians, up to a U(1) phase, in both labs.

The Lorentz invariant
\begin{eqnarray}
I_{2A}=\frac{1}{2} \Tr[\Psi_{AB}^{T}C\Psi_{AB}C\gamma^5],
\end{eqnarray}
is P invariant in Alice's lab but not in Bob's, and is only invariant under evolution generated by zero-mass Dirac Hamiltonians in Alice's lab, up to a U(1) phase, but invariant under unitary evolution generated by arbitrary mass Dirac Hamiltonians in Bob's lab, up to a U(1) phase.

Finally, the Lorentz invariant
\begin{eqnarray}
I_{2B}=\frac{1}{2} \Tr[\Psi_{AB}^{T}C\gamma^5\Psi_{AB}C],
\end{eqnarray}
is invariant under P in Bob's lab but not Alice's lab, and is only invariant under evolution generated by zero-mass Dirac Hamiltonians in Bob's lab, up to a U(1) phase, but invariant under unitary evolution generated by arbitrary mass Dirac Hamiltonians in Alice's lab, up to a U(1) phase.
Each of $I_{2}$, $I_{2A}$ and $I_{2B}$ is identically zero for product states.
Written out in terms of coefficients they are 

\begin{eqnarray}
I_{2}=\psi_{13} \psi_{20}-\psi_{10}\psi_{23} +\psi_{11}\psi_{22}-\psi_{12} \psi_{21}\nonumber\\ +\psi_{02}\psi_{31} -\psi_{01}\psi_{32} +\psi_{00}\psi_{33} -\psi_{03}\psi_{30},\\
I_{2A}=\psi_{00}\psi_{13}-\psi_{03}\psi_{10} +\psi_{02}\psi_{11} -\psi_{01}\psi_{12}\nonumber\\  +\psi_{22}\psi_{31} - \psi_{21}\psi_{32} +\psi_{20}\psi_{33} -\psi_{23}\psi_{30} ,\\
I_{2B}=\psi_{11}\psi_{20} -\psi_{10}\psi_{21} +\psi_{13}\psi_{22} -\psi_{12}\psi_{23}\nonumber\\ +\psi_{00}\psi_{31}-\psi_{01}\psi_{30}  +\psi_{02}\psi_{33}-\psi_{03}\psi_{32}.
\end{eqnarray}
Each of these is a combination of the determinants of four different $2\times 2$ minors of $\Psi_{AB}$.

If both Alice's and Bob's particles are Weyl particles, i.e., if the shared state is invariant under some combination of projections $P_L^A$ or $P_R^A$ by Alice and $P_L^B$ or $P_R^B$ by Bob, each of the Lorentz invariants $I_{1}$, $I_{2}$, $I_{2A}$ and $I_{2B}$ reduces to $4(\psi_{00} \psi_{11}-\psi_{01} \psi_{10})$ up to a sign. This is because the shared state has the symmetry $\psi_{jk}=(-1)^{|LA|}\psi_{(j-2) k}=(-1)^{|LB|}\psi_{j (k-2)}$, where $|LA|=1$ if the state is invariant under $P_L^A$ and zero otherwise, $|LB|=1$ if the state is invariant under $P_L^B$ and zero otherwise, and $j$,$k$ are defined modulo 4. The polynomial $\psi_{00} \psi_{11}-\psi_{01} \psi_{10}$ is the Wootters concurrence \cite{wootters,wootters2}. Thus, for Weyl particles the polynomials $I_{1}$, $I_{2}$, $I_{2A}$ and $I_{2B}$ become essentially equivalent to the Wootters concurrence.

We can construct Lorentz invariants of degree four that are P and CPT invariant in both Alice's and Bob's lab as $I_1^2,I_{2}^2$, $I_{2A}^2$ and $I_{2B}^2$.
The invariants $I_{1}I_{2A}$ and $I_{2}I_{2B}$ are CPT invariant in both labs but only P invariant in Alice's lab,
and $I_{1}I_{2B}$, $I_{2}I_{2A}$ are CPT invariant in both labs but only P invariant in Bob's lab.
Finally, the invariants $I_{1}I_{2}$ and $I_{2A}I_{2B}$ are CPT invariant in both labs but not P invariant in either lab.

A degree four Lorentz invariant can be obtained as

\begin{eqnarray}
I_3&=&-\frac{1}{4} \Tr[\Psi_{AB}^{T}C\Psi_{AB}C\Psi_{AB}^{T}C\Psi_{AB}C]+\frac{1}{2}I_1^2\nonumber\\
&=&-\frac{1}{4} \Tr[\Psi_{AB}^{T}C\gamma^5\Psi_{AB}C\gamma^5\Psi_{AB}^{T}C\gamma^5\Psi_{AB}C\gamma^5]+\frac{1}{2}I_2^2\nonumber\\
&=&\det[\Psi_{AB}].
\end{eqnarray}
This Lorentz invariant is P and CPT invariant in both labs, linearly independent of $I_1^2,I_{2}^2$, $I_{2A}^2$ and $I_{2B}^2$, and identically zero for product states. It is also
invariant under unitary evolution generated by arbitrary mass Dirac Hamiltonians in both labs, up to a U(1) phase. Moreover, it is invariant under any local SU(4) transformations by Alice and Bob, and even any local SL(4,$\mathbb{C}$) transformations $S_{A},S_{B}$ in Alice's or Bob's lab since $\det[S_{A}\Psi_{AB}S_{B}^{T}]=\det[S_{A}]\det[\Psi_{AB}]\det[S_{B}^{T}]=\det[\Psi_{AB}]$.
Written out in terms of coefficients it is
\begin{eqnarray}
I_3=&&(\psi_{01}\psi_{ 30}-\psi_{00}\psi_{ 31})(\psi_{ 13}\psi_{ 22} -\psi_{ 12}\psi_{ 23} ) \nonumber\\&& +(\psi_{01}\psi_{32}-\psi_{ 02}\psi_{ 31})(\psi_{10}\psi_{23} -\psi_{13}\psi_{20} ) \nonumber\\&&+(\psi_{01}\psi_{33}-\psi_{03}\psi_{ 31}) (\psi_{12}\psi_{20 }-\psi_{10}\psi_{22})\nonumber\\&&
  + (\psi_{00}\psi_{32}-\psi_{ 02}   \psi_{ 30})(\psi_{ 13}\psi_{ 21}  - \psi_{11}\psi_{ 23}  ) \nonumber\\&&  +(\psi_{00}\psi_{33 }-\psi_{03} \psi_{30})(\psi_{11}\psi_{22}-\psi_{12}\psi_{21 }) \nonumber\\&&
    +(\psi_{03}\psi_{32}-\psi_{ 02}\psi_{33 })(\psi_{11}\psi_{20} -\psi_{10}\psi_{21}).   
\end{eqnarray}

If both Alice's and Bob's particles are Weyl particles, i.e., if the shared state is invariant under some combination of projections $P_L^A$ or $P_R^A$ by Alice and $P_L^B$ or $P_R^B$ by Bob, the invariant $I_3$ is identically zero since for such states the rank of $\Psi_{AB}$ is at most 2.

Note that by using $C$ and $C\gamma^5$ to construct the Lorentz invariants we have essentially used the spinor representations of the T and CP transformations described in Section \ref{rep}. This follows the same general idea as the construction of the Wootters concurrence in non-relativistic quantum mechanics of using "state inversion" transformations to construct entanglement invariants \cite{wootters2,uhlmann,rungta}.

\subsection{Characterization of spinor entanglement using the invariants }

The five invariants $I_1$, $I_{2}$, $I_{2A}$, $I_{2B}$ and $I_{3}$ are candidate entanglement invariants for a system of two Dirac spinors, invariant under the local proper orthochronous Lorentz groups as well as invariant up to a U(1) phase under local unitary evolutions generated by either zero- or arbitrary mass Dirac Hamiltonians as described in Section \ref{un}.
They provide a partial characterization of the equivalence classes under these local reversible operations.
For a generic state all the five invariants are non-zero. This follows since the zero locus of any non-constant homogeneous polynomial in the state coefficients, i.e., any non-constant homogeneous polynomial on $\mathbb{C}^{16}$, is a lower dimensional subset of $\mathbb{C}^{16}$ (See e.g. \cite{hart} or \cite{wall} Ch. A.1.6.). Thus the ratios of these invariants with homogeneous degree zero can be used to partially characterize inequivalent types of entanglement on almost all of the Hilbert space.
The condition that any of the invariants is zero defines a lower dimensional subset of the Hilbert space with a narrowed range of entanglement properties. Setting additional invariants to zero produces progressively lower dimensional subsets with increasingly constrained range of entanglement properties.

The existence of further Lorentz invariants that are also invariant, up to a U(1) phase, under local unitary evolution generated by either zero- or arbitrary mass Dirac Hamiltonians and are algebraically independent of $I_1, I_{2}$, $I_{2A}, I_{2B}$ and $I_{3}$ has not be ruled out even though none have been found in this work. Thus it has not been ruled out that a more complete characterization of spinor entanglement properties using polynomial Lorentz invariants can be found.
However, any other potentially more complete set of polynomial Lorentz invariants still provides only a partial characterization of the spinor entanglement properties.
To see this we note that there exist spinor entangled states for which $I_1,I_{2}$, $I_{2A},I_{2B}$ and $I_{3}$ are all zero. Thus entangled states exist that cannot be distinguished from product states by any of the five invariants. This is true also for any other set of polynomial Lorentz invariants. Due to the properties of the spinor representation of the Lorentz group there are entangled states for which no homogeneous Lorentz invariant polynomial can be non-zero (See Appendix \ref{locus}).

When the allowed local unitary evolutions of a spinor are generated by Dirac Hamiltonians with nonzero mass the set of local unitary transformations of that spinor that can be implemented forms a dense subset of a Lie group $G^{C\gamma^5}_U$ of unitary transformations. The group $G^{C\gamma^5}_U$ consists of all unitary transformations that preserve the bilinear form $\psi^TC\gamma^5\varphi$ up to a U(1) phase and is isomorphic to $\mathrm{U(1)}\times\mathrm{Sp}(2)$ where $\mathrm{Sp}(2)$ is the compact symplectic group of $4\times 4$ matrices (See e.g. Ref. \cite{hall}  Ch. 1.2.8). Therefore any continuous function that is invariant, up to a U(1) phase, under local unitary evolutions of the spinor generated by Dirac Hamiltonians with nonzero mass is invariant, up to a U(1) phase, under local operations in $G^{C\gamma^5}_U$ on the spinor.
See Appendix \ref{lie} for a discussion.
The group $G^{C\gamma^5}$ of all linear transformations that preserve the bilinear form $\psi^TC\gamma^5\varphi$ up to a U(1) phase is isomorphic to $\mathrm{U(1)}\times\mathrm{Sp}(4,\mathbb{C})$ where $\mathrm{Sp}(4,\mathbb{C})$ is the symplectic group of $4\times 4$ matrices (See e.g. Ref. \cite{hall} Ch. 1.2.4). The group $G^{C\gamma^5}$ is the smallest connected matrix Lie group that contains $G^{C\gamma^5}_U$ and the spinor representation of the proper orthochronous Lorentz group as subgroups. Any continuous Lorentz invariant function that is also invariant, up to a U(1) phase, under local unitary evolutions generated by Dirac Hamiltonians with nonzero mass is invariant, up to a U(1) phase, under local operations in $G^{C\gamma^5}$. See Appendix \ref{lie} for a discussion.

Similarly, when the allowed local unitary evolutions of a spinor are generated by Dirac Hamiltonians with zero mass and a coupling to a Yukawa pseudoscalar boson the set of local unitary transformations that can be implemented forms a dense subset of a Lie group $G^{C}_U$ of unitary transformations. The group $G^{C}_U$ consists of all unitary transformations that preserve the bilinear form $\psi^TC\varphi$ up to a U(1) phase and is isomorphic to $\mathrm{U(1)}\times\mathrm{Sp}(2)$. Therefore any continuous function that is invariant, up to a U(1) phase, under local unitary evolutions of the spinor generated by Dirac Hamiltonians with zero mass and a coupling to a Yukawa pseudoscalar boson is invariant, up to a U(1) phase, under local operations in $G^{C}_U$ on the spinor.
The group $G^{C}$ of all linear transformations that preserve the bilinear form $\psi^TC\varphi$ up to a U(1) phase is isomorphic to $\mathrm{U(1)}\times\mathrm{Sp}(4,\mathbb{C})$. The group $G^{C}$ is the smallest connected matrix Lie group that contains $G^{C}_U$ and the spinor representation of the proper orthochronous Lorentz group as subgroups. Any continuous Lorentz invariant function that is also invariant, up to a U(1) phase, under local unitary evolutions generated by Dirac Hamiltonians with zero mass an a coupling to a Yukawa pseudoscalar boson is invariant, up to a U(1) phase, under local operations in $G^{C}$. See Appendix \ref{lie} for a discussion. 

If the allowed local unitary evolutions of a spinor are generated by Dirac Hamiltonians with zero mass and no additional couplings the set of local unitary transformations that can be implemented forms a dense subset of the group $G^{C}_U\cap G^{C\gamma^5}_U$ of unitary transformations that preserve both the bilinear form $\psi^TC\varphi$ and the bilinear form $\psi^TC\gamma^5\varphi$ up to a U(1) phase, which is isomorphic to $\mathrm{U(1)}\times\mathrm{SU}(2)\times\mathrm{SU}(2)$. Therefore any continuous function that is invariant, up to a U(1) phase, under local unitary evolutions of the spinor generated by Dirac Hamiltonians with zero mass and no additional couplings is invariant, up to a U(1) phase, under local operations in $G^{C}_U\cap G^{C\gamma^5}_U$ on the spinor.
The group $G^{C}\cap G^{C\gamma^5}$ of all linear transformations that preserve both the bilinear form $\psi^TC\varphi$ and the bilinear form $\psi^TC\gamma^5\varphi$ up to a U(1) phase is isomorphic to $\mathrm{U}(1)\times\mathrm{SL}(2,\mathbb{C})\times \mathrm{SL}(2,\mathbb{C})$. The group $G^{C}\cap G^{C\gamma^5}$ is the smallest connected matrix Lie group that contains $G^{C}_U\cap G^{C\gamma^5}_U$ and the spinor representation of the proper orthochronous Lorentz group as subgroups. Any continuous Lorentz invariant function that is also invariant, up to a U(1) phase, under local unitary evolutions generated by Dirac Hamiltonians with zero mass and no additional couplings is invariant, up to a U(1) phase, under local operations in $G^{C}\cap G^{C\gamma^5}$. See Appendix \ref{lie} for a discussion.

The Lorentz invariant $|I_1|$ is invariant under $G^{C}\otimes G^{C}$, while  $|I_2|$ is invariant under $G^{C\gamma^5}\otimes G^{C\gamma^5}$, $|I_{2A}|$ is invariant under $G^{C}\otimes G^{C\gamma^5}$ and $|I_{2B}|$ is invariant under $G^{C\gamma^5}\otimes G^{C}$. The Lorentz invariant $|I_3|$ is invariant under $\mathrm{U}(1)\times\mathrm{SL}(4,\mathbb{C})\otimes\mathrm{U}(1)\times \mathrm{SL}(4,\mathbb{C})$ and thus invariant under $G^{C}\otimes G^{C}$, $G^{C\gamma^5}\otimes G^{C\gamma^5}$, $G^{C}\otimes G^{C\gamma^5}$ and $G^{C\gamma^5}\otimes G^{C}$. A subset of the Lorentz invariants that are invariant up to a U(1) phase under a shared group can be use to partially characterize orbits of this shared group. For example the subset  $I_1$ and $I_3$ can be used to partially characterize orbits of $G^{C}\otimes G^{C}$ and the subset $I_1,I_{2}$, $I_{2A},I_{2B}$, $I_{3}$ can be used to partially characterize orbits of $G^{C}\cap G^{C\gamma^5}\otimes G^{C}\cap G^{C\gamma^5}$.
But a subset the Lorentz invariants that are invariant up to a U(1) phase for a shared group cannot distinguish between two orbits of reversible operations that are both contained in the same orbit of this shared group.

The absolute values of the Lorentz invariants $|I_1|$, $|I_{2}|$, $|I_{2A}|$, $|I_{2B}|$ and $|I_{3}|$ can be used to construct Lorentz invariants also for states that are incoherent mixtures, i.e., mixed states. Mixed states are represented by positive semi-definite Hermitian matrices that are the analogues of density matrices. Each absolute value of a Lorentz invariant can be extended to a Lorentz invariant on the set of such density matrices through a {\it convex roof extension} \cite{lima,wakker,uhlmannn} (See Appendix \ref{con} for a discussion). Any such convex roof extension is identically zero on the set of separable states, i.e., the set of states that are incoherent mixtures of product states.

We can consider a set of incoherent mixtures of states with the same definite momentum degrees of freedom.
In this case the convex roof extension of a Lorentz invariant is invariant under the same local unitary evolutions generated by Dirac Hamiltonians as the absolute value of the Lorentz invariant it is derived from. 
Since the Lorentz invariants $|I_1|,|I_{2}|$, $|I_{2A}|,|I_{2B}|$ and $|I_{3}|$ have only been defined for definite momenta the case of incoherent mixtures of states with different momenta cannot be treated this way.
In particular we can obtain incoherent mixtures from a pure state without definite momenta by partially tracing over the momentum degrees of freedom. Such reduced density matrices do not have well defined transformation properties under local unitary evolution generated by Dirac Hamiltonians and the convex roof extensions 
of $|I_1|,|I_{2}|$, $|I_{2A}|,|I_{2B}|$ and $|I_{3}|$ are in general not invariant (See Appendices \ref{dwalin} and \ref{con} for a discussion).

From the matrix $\Psi_{AB}$ one can construct the analogue of the one-party reduced density matrix for Alice's particle $ \Psi_{AB} \Psi_{AB}^\dagger$, and likewise the one-party reduced density matrix $\Psi_{AB}^T \Psi_{AB}^*$ for Bob's particle.  The rank and determinant of the these reduced density matrices are invariant under the spinor representation of the local proper orthochronous Lorentz groups, as well as under local unitary evolution generated by Dirac Hamiltonians. However, in general their eigenvalues are not invariant. Other one-party reduced matrices with eigenvalues that are invariant under the spinor representations of the local proper orthochronous Lorentz groups, as well as invariant, up to a U(1) phase, under local unitary evolution generated by either zero- or arbitrary mass Dirac Hamiltonians can be constructed. See Appendix \ref{gloin} for a discussion.

For a local unitary evolution generated by a time independent nonzero-mass Dirac Hamiltonian on Alice's side the absolute values of $I_1$ and $I_{2A}$ oscillate with a constant time average. For zero momentum and four-potential the angular frequency of oscillation is $2m_A$ where $m_A$ is the mass of Alice's particle. Likewise, for a local unitary evolution generated by a time independent nonzero-mass Dirac Hamiltonian on Bobs's side the absolute values of $I_1$ and $I_{2B}$ oscillate, and for zero momentum and four-potential the angular frequency of oscillation is $2m_B$ where $m_B$ is the mass of Bob's particle.
The angular frequencies $2m_A$ and $2m_B$ are the same as those of the Zitterbewegung \cite{breit} of the respective Dirac particle. See Appendix \ref{ozz} for details.

In the context of Dirac particles in 2D and 3D Dirac semimetals we can see from the discussion in Section \ref{ham} that 
the absolute values of  $I_1,I_{2}$, $I_{2A},I_{2B}$ and $I_{3}$ are invariant under local unitary evolution generated by  Hamiltonians of the form given in Eq. (\ref{2d}) for Dirac particles in 2D Dirac semimetals \cite{kotov,vass}, and  Hamiltonians on the form given in Eq. (\ref{3d}) for Dirac particles in 3D Dirac semimetals \cite{bohm}. Moreover, we can add 
Semenoff mass terms $M_S\gamma^0\gamma^3$ \cite{semenoff} and Haldane mass terms $M_{H}\gamma^5\gamma^0\gamma^3$ \cite{haldane} and still maintain the invariance of the absolute values of  $I_1,I_{2}$, $I_{2A},I_{2B}$ and $I_{3}$.

\section{Examples of spinor entangled states}\label{exx}

Here we consider a few examples of entangled states to illustrate how the quantities $I_1,I_{2}$, $I_{2A},I_{2B}$ and $I_{3}$ distinguish them.
Reference \cite{pachos} studied the generation of entanglement between the spinorial degrees of freedom of two Dirac particles. In particular it considered the so called spinor "EPR-state" $1/\sqrt{2}({\phi_0^A}\otimes{\phi_1^B}-i{\phi_1^A}\otimes{\phi_0^B})$. For this state only $I_1$ is non-zero and it attains the absolute value $1/2$ which is its maximal absolute value for normalized states. Since $I_1$ is Lorentz invariant it attains the same absolute value for all states related to this state by spinor representations of local Lorentz transformations. Assuming zero mass for both particles it attains the absolute value $1/2$ also on all states related to the spinor EPR-state by unitary evolution generated locally by zero-mass Dirac Hamiltonians. The same state or equivalent states were considered also in \cite{alsing,mano,moradi,geng}.

In a similar way we can construct a state $1/\sqrt{2}({\phi_1^A}\otimes{\phi_3^B}-{\phi_2^A}\otimes{\phi_0^B})$ for which only $I_2$ is non-zero and attains the absolute value $1/2$, its maximal absolute value for normalized states. The state $1/\sqrt{2}({\phi_0^A}\otimes{\phi_0^B}-{\phi_1^A}\otimes{\phi_3^B})$ is such that only $I_{2A}$ is non-zero and attains the absolute value $1/2$, its maximal absolute value for normalized states. The state $1/\sqrt{2}({\phi_1^A}\otimes{\phi_1^B}-{\phi_2^A}\otimes{\phi_0^B})$ is such that only $I_{2B}$ is non-zero and attains the absolute value $1/2$, its maximal absolute value for normalized states. An example of a state such that $I_1=I_{2}=I_{2A}=I_{2B}=0$ but $I_{3}$ attains the absolute value $1/16$, its maximal absolute value for normalized states, is $1/2({\phi_0^A}\otimes{\phi_1^B}+{\phi_3^A}\otimes{\phi_0^B}+i{\phi_2^A}\otimes{\phi_3^B}+i{\phi_1^A}\otimes{\phi_2^B})$.

We may also consider chiral versions of the EPR state constructed from spinors with definite chirality. For example, with right-handed chirality of both Alice'a and Bob's spinors a chiral EPR state is
$1/\sqrt{8}[({\phi_1^A}+{\phi_3^A})\otimes({\phi_0^B}+{\phi_2^B})-i({\phi_0^A}+{\phi_2^A})\otimes({\phi_1^B}+{\phi_3^B})]$. For this state the Lorentz invariants $I_1$, $I_2$, $I_{2A}$ and $I_{2B}$ all attain their maximum possible absolute values over the normalized states while $I_3=0$.

Entanglement as the result of a decay process of a particle with zero spin into a particle-antiparticle pair was considered in Ref. \cite{caban}. An example of a state that can be given an interpretation as the result of such a decay process is $1/\sqrt{2}({\phi_0^A}\otimes{\phi_3^B}-i{\phi_1^A}\otimes{\phi_2^B})$.
For this state only $I_1$ is non-zero and attains the absolute value $1/2$, its maximal absolute value for normalized states. 
Another state that can similarly be given such an interpretation but which is symmetric with respect to Alice and Bob is $1/2({\phi_0^A}\otimes{\phi_3^B}+{\phi_3^A}\otimes{\phi_0^B}-{\phi_1^A}\otimes{\phi_2^B}-{\phi_2^A}\otimes{\phi_1^B})$. For this state $I_{2A}=I_{2B}=0$ but $I_1,I_{2}$ and $I_{3}$ attain their respective maximal absolute values over the normalized states.

\section{The Foldy-Wouthuysen representation}\label{wout}

The Pauli equation was originally introduced in Ref. \cite{pauli2} to describe a non-relativistic spin-$\frac{1}{2}$ particle. For a particle with mass $m$ and charge $q$ in an electromagnetic four-potential $A_{\mu}(x)$ it can be written, in natural units $\hbar=c=1$, as
\begin{eqnarray}\label{p}
i\partial_0\chi&=&\Bigg[ qA_0 I +m I+\frac{1}{2m}\sum_{\mu=1,2,3}(i\partial_\mu-qA_\mu)^2I\nonumber\\
 &&-\frac{iq}{2m}\sum_{\substack{\mu,\nu=1,2,3\\\mu\neq\nu}}\sigma^\mu\sigma^\nu\partial_{\mu}A_\nu\Bigg]\chi,
\end{eqnarray}
where $\chi$ is a two component spinor.
We identify the Pauli Hamiltonian as
\begin{eqnarray}
H_{P}&=& qA_0 I +m I+\frac{1}{2m}\sum_{\mu=1,2,3}(i\partial_\mu-qA_\mu)^2I\nonumber\\
 &&-\frac{iq}{2m}\sum_{\substack{\mu,\nu=1,2,3\\\mu\neq\nu}}\sigma^\mu\sigma^\nu\partial_{\mu}A_\nu.
\end{eqnarray}
For two spacelike separated Pauli spinors undergoing evolution generated locally by Pauli Hamiltonians and acting unitarily on a subspace spanned by the spin degrees of freedom the invariant of the evolution describing the spin entanglement is the Wootters concurrence \cite{wootters,wootters2}.

Since the Pauli equation is used to describe non-relativistic spin-$\frac{1}{2}$ particles it is expected that it is related to an approximation of the Dirac equation valid for low momenta and weak fields. Moreover, since the Pauli Hamiltonian acts only on a two component spinor it would be required for this case that a transformation can take the Dirac Hamiltonian to a block diagonal form with two $2\times2$ blocks, to a good approximation.
A method to recover the Pauli Hamiltonian as an approximation to the Dirac Hamiltonian in the low momentum weak field limit is the Foldy-Wouthuysen transformations \cite{foldy}. This procedure creates from the Dirac Hamiltonian through a series of canonical transformations a Hamiltonian that is a series expansion in $1/m$. Each canonical transformation eliminates the Hamiltonian terms that are not on the block diagonal form to one higher order in $1/m$. The generators of such canonical transformations are in general functions of the momentum and the four-potential and its derivatives. In each step the state and the Hamiltonian transform according to

\begin{eqnarray}\label{dc}
\psi_{(n)}&=&e^{iS_{n}}\psi_{(n-1)},\nonumber\\
H_{(n)}&=&e^{iS_{n}}H_{(n-1)}e^{-iS_{n}}-ie^{iS_{n}}\partial_t e^{-iS_{n}},
\end{eqnarray}
where $S_{n}$ is the generator of the $n$th transformation and $\psi_{(n)}$ and $H_{(n)}$ are the resulting spinor and Hamiltonian, respectively.
The first two canonical transformations in the Foldy-Wouthuysen procedure are generated by 
\begin{eqnarray}\label{gen}
S_1&=&\frac{i}{2m}\sum_{\mu=1,2,3}\gamma^\mu(i\partial_\mu-qA_\mu),\nonumber\\
S_2&=&\frac{q}{4m^2}\gamma^0\sum_{\mu=1,2,3}\gamma^\mu(\partial_\mu A_0+\partial_0A_\mu).
\end{eqnarray}
These two transformations result in the Hamiltonian
\begin{eqnarray}\label{fw}
H_{FW(2)}&=& qA_0 I +m\gamma^0
+\frac{1}{2m}\gamma^0\sum_{\mu=1,2,3}(i\partial_\mu-qA_\mu)^2\nonumber\\
&-&\frac{iq}{2m}\gamma^0\sum_{\substack{\mu,\nu=1,2,3\\\mu\neq\nu}}\gamma^\mu\gamma^\nu\partial_{\nu}A_\mu+\mathcal{O}\left(\frac{1}{m^2}\right).
\end{eqnarray}
To first order in $1/m$ the Hamiltonian $H_{FW(2)}$ can be decoupled into two Pauli Hamiltonians acting on the upper two and lower two components of the spinor, respectively (See Ref. \cite{foldy} for details). The third transformation adds the spin-orbit coupling and the Darwin term and subsequent transformations add additional relativistic corrections to the two Pauli Hamiltonians.

Now we consider the case of a fixed momentum subspace, i.e., a subspace spanned by spinorial degrees of freedom, and how the bilinear forms $\psi^TC\varphi$ and $\psi^TC\gamma^5\varphi$ on such a space are represented in the Foldy-Wouthuysen picture. From Eq. (\ref{gen}) together with Eq. (\ref{c5}) we can see that $C\gamma^5 S_1=-S_1^TC\gamma^5$ and $C\gamma^5 S_2=-S_2^TC\gamma^5$.  From Eq. (\ref{gen}) together with Eq. (\ref{c}) we can see that $C S_1=S_1^TC$ and $C S_2=-S_2^TC$.
Therefore, $\psi^TC\gamma^5\varphi=\psi^T_{(1)}C\gamma^5\varphi_{(1)}=\psi^T_{(2)}C\gamma^5\varphi_{(2)}$, but $\psi^TC\varphi=\psi^T_{(1)}Ce^{-2iS_1}\varphi_{(1)}=\psi^T_{(2)}Ce^{iS_{2}} e^{-2iS_1}e^{-iS_2}\varphi_{(2)}$.

Next we define $\Psi_{AB}^{(n)}=e^{iS^A_n}\dots e^{iS^A_1}\Psi_{AB}(e^{iS^{B}_1})^T\dots (e^{iS^{B}_n})^T$ where $S^A_i$ is the generator of the $i$th transformation on Alice's side and $S^B_i$  is the generator of the $i$th transformation on Bob's side.
The representation of $I_1$ in the second step Foldy-Wouthuysen picture is in general a function of $\Psi_{AB}^{(2)}$, the momenta, the four-potentials and their derivatives 

\begin{eqnarray}
I_1=\frac{1}{2}\Tr[\Psi_{AB}^{(2)T}Ce^{iS^A_{2}}e^{-2iS^A_1}e^{-iS^A_2}\Psi_{AB}^{(2)}C e^{iS^B_{2}}e^{-2iS^B_1}e^{-iS^B_2}].
\end{eqnarray}
The representation of $I_2$ on the other hand is still a polynomial in the state coefficients
\begin{eqnarray}
I_{2}=\frac{1}{2} \Tr[\Psi_{AB}^{(2)T}C\gamma^5\Psi_{AB}^{(2)}C\gamma^5].
\end{eqnarray}
Like $I_1$ the representations of $I_{2A}$ and $I_{2B}$ are functions of the momenta, the four-potentials and their derivatives
\begin{eqnarray}
I_{2A}=\frac{1}{2} \Tr[\Psi_{AB}^{(2)T}Ce^{iS^A_{2}} e^{-2iS^A_1}e^{-iS^A_2}\Psi_{AB}^{(2)}C\gamma^5],
\end{eqnarray}
\begin{eqnarray}
I_{2B}=\frac{1}{2} \Tr[\Psi_{AB}^{(2)T}C\gamma^5\Psi_{AB}^{(2)}Ce^{iS^B_{2}} e^{-2iS^B_1}e^{-iS^B_2}].
\end{eqnarray}
However, $I_3$ is still a polynomial in the state coefficients
\begin{eqnarray}
I_3=\det[\Psi_{AB}^{(2)}],
\end{eqnarray}
since it is a determinant and the Foldy-Wouthuysen transformations are unitary matrices with the first two being determinant one.

\section{Discussion and Conclusions}\label{diss}
We have considered the problem of describing the spinor entanglement of two Dirac particles with definite momenta held by spacelike separated observers Alice and Bob. We reviewed some properties of the Dirac equation, the spinor representation of the Lorentz group and the charge conjugation, and discussed properties of Lorentz invariant bilinear forms. 
The assumption was made that we can neglect curvature and describe both Alice and Bob as being in a Minkowski space. Further, we assumed that it is in some way physically motivated to use a description where particle momentum eigenmodes have a finite spatial extent. 
Lastly, we assumed that the tensor products of the single particle momentum eigenmodes is a basis for the two-particle states.

Given these assumptions and using the properties of the Lorentz invariant bilinear forms we have constructed five polynomials $I_1,I_2,I_{2A},I_{2B},$ and $I_3$ in the state coefficients of the two spacelike separated Dirac particles, that are invariant under the spinor representations of the local proper orthochronous Lorentz groups. The invariants, $I_1,I_2,I_{2A},I_{2B}$, are of degree 2 and $I_3$ is of degree 4. Each of these Lorentz invariants is identically zero for all product states but does not take the value zero for all other states. 

The four Lorentz invariants of degree 2 can each be expressed as a sum of four determinants of $2\times 2$ matrices of state coefficients. For the case of two particles with definite chirality, i.e., Weyl particles, these invariants reduce to a single determinant, the Wootters concurrence \cite{wootters,wootters2}.

We considered evolutions that are generated by local Hamiltonians and act unitarily on subspaces with fixed momenta, i.e., subspaces spanned by the spinorial degrees of freedom.
The polynomial $I_2$ is invariant, up to a U(1) phase, under such local unitary evolution generated locally by Dirac Hamiltonians. The polynomial $I_1$ is invariant, up to a U(1) phase, only for zero-mass Dirac Hamiltonians.
The polynomial $I_{2A}$ is invariant, up to a U(1) phase, for arbitrary-mass Dirac Hamiltonians acting on Bobs side but only for zero-mass Dirac Hamiltonians acting on Alice's side.
Similarly, $I_{2B}$ is invariant, up to a U(1) phase, for arbitrary-mass Dirac Hamiltonians acting on Alice's side but only for zero-mass Dirac Hamiltonians acting on Bob's side.
The fifth Lorentz invariant, $I_3$, is invariant, up to a U(1) phase, under any local unitary evolution, physically allowed or not.

For a system of two Dirac particles with definite momenta the
conditions of non-existence for product states, invariance under local evolutions generated by physically allowed Dirac Hamiltonians that act unitarily on any subspace defined by fixed particle momenta, and Lorentz invariance were tentatively considered as the conditions defining a spinor entanglement property. With this definition two spinor entangled states that can be transformed into each other by physically allowed local unitary evolution and changes of reference frame have the same entanglement properties.
The polynomials $I_1,I_2,I_{2A},I_{2B},$ and $I_3$ can be used to partially characterize such spinor entanglement properties, i.e., partially characterize the qualitatively different ways that two Dirac spinors can be entangled.

For the case of incoherent mixtures of spinor entangled states the Lorentz invariants $|I_1|,|I_2|,|I_{2A}|,|I_{2B}|,$ and $|I_3|$ can be extended to Lorentz invariant functions on the set of such incoherent mixtures through convex roof extensions \cite{lima,wakker,uhlmannn}.
Thus these convex roof extensions provide a partial characterization of the qualitatively different types of spinor entanglement of incoherent mixtures. All such convex roof extensions are by definition identically zero for the separable states, i.e., for the incoherent mixtures of product states.

The constructed polynomials were considered also in the context of Dirac and Weyl quasiparticles in condensed matter and molecular systems. In particular, they are all invariant, up to a U(1) phase, for local evolution generated by the Hamiltonians describing a Dirac particle in the 2D Dirac semimetal graphene \cite{kotov,vass} and the Hamiltonians describing a particle in 3D Dirac semimetals \cite{bohm}. This holds also if Semenoff \cite{semenoff} or Haldane \cite{haldane} mass terms are added to the Hamiltonians.

We considered examples of spinor entangled states characterized by the Lorentz invariants. For each invariant there exist states for which only that invariant is nonzero. 
Only $I_1$ is non-zero for the so called spinor "EPR-state" previously discussed in the literature \cite{pachos,alsing,mano,moradi,geng}.

If the local evolution on Alice's side is generated by a time independent nonzero-mass Dirac Hamiltonian the absolute values of $I_1$ and $I_{2A}$ oscillate with constant time average. The frequency of oscillation is the same as that of the Zitterbewegung \cite{breit} of Alice's Dirac particle.
Likewise, for a local evolution on Bobs's side generated by a time independent nonzero-mass Dirac Hamiltonian the absolute values of $I_1$ and $I_{2B}$ oscillate with the Zitterbewegung frequency of Bob's Dirac particle.

Finally, we considered the Foldy-Wouthuysen representations of the polynomial invariants. It was found that in general the invariants $I_1,I_{2A}$, and $I_{2B}$ become functions of the momenta and four-potentials and their derivatives after the two first Foldy-Wouthuysen transformations. Only $I_2$ and $I_3$ are still purely polynomials in state coefficients. 

Several previous works \cite{czachor,caban,caban3,moradi,caban2,terno2,ahn,tera,tera2,won} have investigated entanglement of two Dirac particles by studying non-local correlations  and in particular the violation of Bell-inequalities \cite{bell,chsh}.
This approach requires the identification of appropriate Bell inequalities together with choices of measurement operators that allow non-local correlations to be observed. In contrast the approach in this work constructs algebraic quantities from the state coefficients and does not require identification of measurement operators, but also does not directly demonstrate the potential for non-local correlations.

Many previous works have investigated entanglement of two Dirac particles in the context of a Quantum Field Theory formalism \cite{czachor,alsing,pachos,mano,caban,caban3,leon,tessier,geng,caban2,terno2,ahn,tera,tera2}
including non-local correlations between Dirac particles \cite{czachor,caban,caban3,caban2,terno2,ahn,tera,tera2}, generation of entanglement \cite{pachos,mano,leon}, and entanglement in non-inertial frames \cite{tessier}.
In a Quantum Field Theory formalism the Dirac spinor is reinterpreted as an operator valued Dirac field acting on a Hilbert space. In the context of such a formalism the Lorentz invariants constructed in this work can be reinterpreted as Lorentz invariants of the Dirac field since this field still transforms under the spinor representation of the Lorentz group.
 The Hilbert space in a Quantum Field Theory formalism on the other hand is typically constructed to transform under an infinite dimensional representation of the Lorentz group. The Hilbert space basis vectors are labelled by momentum and a spin degree of freedom chosen so that the infinite dimensional representation of the Lorentz group acts on the spin degree of freedom conditioned on the particle momentum. Since this spin degree of freedom transforms conditioned on the momentum it does not have a complete physical interpretation independently of the momentum. 
 
One consequence of the spin degree of freedom in a Quantum Field Theory formalism being defined in a momentum dependent way is that taking the partial trace over the momentum in general leads to reduced spin density matrices that are not physically meaningful and do not transform under any representation of the Lorentz group \cite{terno}. Only in the case of definite particle momentum can the reduced spin density matrix be given a physical meaning and has well defined transformation properties. For the same reason spin entanglement between two particles in this kind of formalism can in general not be meaningfully described independently of the particle momenta. If one nevertheless constructs the reduced two particle spin density matrix one finds that the
mathematical counterpart of spin entanglement is in general dependent on the inertial frame and does not have well defined transformation properties \cite{adami}. See Appendix \ref{kili} for a discussion of the infinite dimensional representation of the Lorentz group.
In contrast the action of the spinor representation of the Lorentz group is not conditioned on the momentum and therefore physically meaningful reduced density matrices for Dirac spinors can be constructed by taking the partial trace over the momenta. These transform under the spinor representations of the local proper orthochronous Lorentz groups. Moreover, the Lorentz invariants $|I_1|,|I_2|,|I_{2A}|,|I_{2B}|,$ and $|I_3|$ can be extended to Lorentz invariant functions on the set of such reduced two-spinor density matrices through convex roof extensions \cite{lima,wakker,uhlmannn} (See Appendix \ref{con} for a discussion).

The five Lorentz invariants in this work were constructed to describe spinor entanglement for the case of definite particle momenta but whether similar constructions can be made for the case of Dirac particles without definite momenta is an open question.
Another open question is that of describing quantifiable spinor entanglement properties, i.e., spinor entanglement properties that satisfy a condition of non-increase on average under any local operations assisted by classical communication \cite{vidal}.

\begin{acknowledgments}
The author thanks Antonio Ac{\'i}n for comments and discussion and for suggestions that encouraged the addition of Appendix G. The author also thanks the anonymous referees for comments that prompted several additions including Appendices B and D-H.  Support
from the European Research Council Consolidator Grant QITBOX (Grant Agreement No. 617337), the Spanish MINECO
(Project FOQUS FIS2013-46768-P, Severo Ochoa grant SEV-
2015-0522), Fundaci\'o Privada Cellex, the Generalitat de Catalunya
(SGR 875) and the John Templeton Foundation is acknowledged.

\end{acknowledgments}

\appendix

\section{Real and rational numbers as quantifiers in experiments}\label{opp}
A measurement quantifying a property is an operational procedure that terminates and results in an output number that is registered by the experimenter. Only finitely many such procedures can be performed in any given experiment and the registry holding the output numbers has finite capacity.

A real number is defined as the limit of a Cauchy sequence of rational numbers \cite{cantor}. If the sequence terminates the limit is a rational number itself. If the limit is not a rational number the sequence does not terminate.
Formulated differently, an irrational number is represented in a base-$n$ positional numeral system  as a non-terminating and non-recurring sequence of digits for any $n$.

Therefore, no finite capacity registry can hold an irrational number, and thus the output numbers from any experiment is necessarily a finite set of rational numbers. The elements of a finite set of rational numbers are multiples of their greatest common divisor $q$. 
Thus, the experiment cannot distinguish between a continuous spectrum $\mathbb{R}$ and the discrete spectrum $nq$, $n\in \mathbb{Z}$.
Hence, for any experimental setup it is possible to use a model where the measurable quantities take only a discrete set of values.

\section{The case without definite particle momenta}\label{dwalin}

The assumption of fixed particle momenta was made so that spinor entanglement could be studied without involving momentum degrees of freedom. However, this assumption is not physically trivial and as described in Sect. \ref{dir} it can be made in some models, such as those using box quantization, but in others not. In a model using a rigged Hilbert space description it can only be done as an approximation. Moreover, even if a given model allows for fixed momenta one may wish to consider physical scenarios without this restriction.

Therefore we consider the qualitative description of spinor entanglement for the case where the momentum is not fixed. A general initial state $\psi_{AB}$ with multiple momentum components can be written as

\begin{eqnarray}
\psi_{AB}=\sum_{\bold{k_A},\bold{k_B}}\sum_{j_A,j_B}\psi_{j_A,j_B,\bold{k_A},\bold{k_B}}\phi_{j_A}e^{i\bold{k_A}\cdot\bold{x_A}}\otimes \phi_{j_B}e^{i\bold{k_B}\cdot\bold{x_B}}.
\end{eqnarray}
Since the Dirac Hamiltonian $H_D$ has a term $-i\sum_{\mu=1,2,3}\gamma^0\gamma^\mu \partial_{\mu}$ that contains derivatives with respect to the spatial coordinate it is clear that the evolution of Alice's spinor is conditioned on the initial particle momentum $\bold{k_A}$.
Therefore the evolution of the spinorial degrees of freedom is described by a set of unitaries $U(\bold{k_A},t)$ indexed by $\bold{k_A}$ and acting on the terms of the state with the corresponding momentum, where in general $U(\bold{k_A},t)\neq U(\bold{k_A}',t)$ if $\bold{k_A}\neq \bold{k_A}'$. Likewise, the evolution of Bob's particle is described by a set of unitaries $U(\bold{k_B},t)$ indexed by $\bold{k_B}$ where in general $U(\bold{k_B},t)\neq U(\bold{k_B}',t)$ if $\bold{k_B}\neq \bold{k_B}'$.

Let us for simplicity consider an evolution that preserves the subspaces defined by fixed momentum.
A state initially without entanglement in the momentum degrees of freedom but without definite particle momenta such as

\begin{eqnarray}
\psi_{AB}=\frac{1}{\sqrt{8}}[&&\phi_{1}(e^{i\bold{k_A^1}\cdot\bold{x_A}}+e^{i\bold{k_A^2}\cdot\bold{x_A}})\otimes \phi_{1}(e^{i\bold{k_B^1}\cdot\bold{x_B}}+e^{i\bold{k_B^2}\cdot\bold{x_B}})\nonumber\\
&&+\phi_{0}(e^{i\bold{k_A^1}\cdot\bold{x_A}}+e^{i\bold{k_A^2}\cdot\bold{x_A}})\otimes \phi_{0}(e^{i\bold{k_B^1}\cdot\bold{x_B}}+e^{i\bold{k_B^2}\cdot\bold{x_B}})],\nonumber\\
\end{eqnarray}
is evolved conditioned on the initial momenta to a state
\begin{eqnarray}
\psi_{AB}(t)=\frac{1}{\sqrt{8}}[&&(U(\bold{k_A^1},t)\phi_{1}e^{i\bold{k_A^1}\cdot\bold{x_A}}+U(\bold{k_A^2},t)\phi_{1}e^{i\bold{k_A^2}\cdot\bold{x_A}})\nonumber\\
&&\otimes (U(\bold{k_B^1},t)\phi_{1}e^{i\bold{k_B^1}\cdot\bold{x_B}}+U(\bold{k_B^2},t)\phi_{1}e^{i\bold{k_B^2}\cdot\bold{x_B}})\nonumber\\
&&+(U(\bold{k_A^1},t)\phi_{0}e^{i\bold{k_A^1}\cdot\bold{x_A}}+U(\bold{k_A^2},t)\phi_{0}e^{i\bold{k_A^2}\cdot\bold{x_A}})\nonumber\\&&\otimes (U(\bold{k_B^1},t)\phi_{0}e^{i\bold{k_B^1}\cdot\bold{x_B}}+U(\bold{k_B^2},t)\phi_{0}e^{i\bold{k_B^2}\cdot\bold{x_B}})],\nonumber\\
\end{eqnarray}
where the spinor states are in general conditioned on the momenta. Thus, in the case without definite particle momenta the spinor entanglement cannot be described on its own without involving the momentum degrees of freedom.

\section{The constant momenta and four-potentials limit of $I_1, I_{2A}$, and $I_{2B}$}\label{ozz}
The absolute value of the Lorentz invariant $I_1$ is not invariant under unitary evolution generated by local nonzero-mass Dirac Hamiltonians in Alice's lab or in Bob's lab. The absolute values of the Lorentz invariants $I_{2A}$ and $I_{2B}$ are only invariant under unitary evolution generated by local nonzero-mass Dirac Hamiltonians in Bobs's and in Alice's lab, respectively. Moreover, $I_1$ is the only Lorentz invariant found in this work that is non-zero for the spinor "EPR-state" considered in Refs. \cite{pachos,alsing,mano,geng,moradi}.

We therefore study the behaviour of $I_1$, $I_{2A}$ and $I_{2B}$ when the evolution is generated in both labs by nonzero-mass Dirac Hamiltonians. In particular we consider their time dependencies relative to fixed inertial frames of Alice and Bob, i.e., relative to Alice's time $t_A$ and Bob's time $t_B$.

If all momenta and four-potentials are zero the evolution is generated by the Hamiltonians $H_A=m_A\gamma^0$ and $H_B=m_B\gamma^0$ in Alice's and Bob's labs respectively. It the state at $t_A=0$ and $t_B=0$ is $\Psi_{AB}^{0}$ the time dependencies of the Lorentz invariant $I_1$ given by

\begin{eqnarray}\label{osc}
I_1(t_A,t_B)&=&\frac{1}{2}\Tr[\Psi_{AB}^{T}C\Psi_{AB}C]\nonumber\\&=&\frac{1}{2}\Tr[\Psi_{AB}^{0T}Ce^{-2im_A\gamma^0t_A}\Psi_{AB}^{0}e^{-2im_B\gamma^0t_B}C]\nonumber\\&=&e^{-2i(m_At_A+m_Bt_B)}(\psi^0_{00} \psi^0_{11}-\psi^0_{01} \psi^0_{10})\nonumber\\&+& e^{-2i(m_At_A-m_Bt_B)}(  \psi^0_{02}\psi^0_{13}- \psi^0_{03} \psi^0_{12})\nonumber\\&+& e^{2i(m_At_A-m_Bt_B)}( \psi^0_{20}      \psi^0_{ 31}-\psi^0_{21}\psi^0_{30})\nonumber\\&+&e^{2i(m_At_A+m_Bt_B)}(\psi^0_{22 }\psi^0_{33} - 
 \psi^0_{23}\psi^0_{32}).
\end{eqnarray}
Thus, in the fixed inertial frames of Alice and Bob the value of $I_1$ is periodic in $t_A$ with period $\pi/m_{A}$ and periodic in $t_B$ with period $\pi/m_{B}$. So while the absolute value of $I_1$ is not invariant its $t_A$ average and $t_B$ average are invariant in this limit. Note in particular that for each of the constituent determinants, e.g. $e^{-2i(m_At_A+m_Bt_B)}(\psi^0_{00} \psi^0_{11}-\psi^0_{01} \psi^0_{10})$, the absolute value is constant in $t_A$ and $t_B$. Thus, in this limit the absolute value of the Wootters concurrence of the respective spinor components is invariant.
Note further that a nonzero constant scalar potential $A_0$ does not affect the absolute value of $I_1$ or the constituent concurrences since it only generates a U(1) phase.

Similarly, for $I_{2A}$ we have
\begin{eqnarray}
I_{2A}(t_A)&=&\frac{1}{2}\Tr[\Psi_{AB}^{T}C\Psi_{AB}C\gamma^5]\nonumber\\&=&\frac{1}{2}\Tr[\Psi_{AB}^{0T}e^{-im_A\gamma^0t_A}Ce^{-im_A\gamma^0t_A}\Psi_{AB}^{0}C\gamma^5]\nonumber\\&=&e^{-2im_At_A}(\psi^0_{00} \psi^0_{13}-\psi^0_{03} \psi^0_{10})\nonumber\\&+& e^{-i2m_At_A}(  \psi^0_{02}\psi^0_{11}- \psi^0_{01} \psi^0_{12})\nonumber\\&+& e^{2im_At_A}( \psi^0_{22}\psi^0_{ 31}-\psi^0_{21}\psi^0_{32})\nonumber\\&+&e^{2im_At_A}(\psi^0_{20 }\psi^0_{33} - 
 \psi^0_{23}\psi^0_{30}),
\end{eqnarray}
which is periodic in $t_A$ with a period $\pi/m_A$, and for $I_{2B}$ we have
\begin{eqnarray}
I_{2B}(t_B)&=&\frac{1}{2}\Tr[\Psi_{AB}^{T}C\gamma^5\Psi_{AB}C]\nonumber\\&=&\frac{1}{2}\Tr[\Psi_{AB}^{0T}C\gamma^5\Psi_{AB}^{0}e^{-im_B\gamma^0t_B}Ce^{-im_B\gamma^0t_B}]\nonumber\\&=&e^{-2im_Bt_B}(\psi^0_{11} \psi^0_{20}-\psi^0_{10} \psi^0_{21})\nonumber\\&+& e^{-2im_Bt_B}(  \psi^0_{00}\psi^0_{31}- \psi^0_{01} \psi^0_{30})\nonumber\\&+& e^{2im_Bt_B}( \psi^0_{13}\psi^0_{ 22}-\psi^0_{12}\psi^0_{23})\nonumber\\&+&e^{2im_Bt_B}(\psi^0_{02 }\psi^0_{33} - 
 \psi^0_{03}\psi^0_{32}),
\end{eqnarray}
which is periodic in $t_B$ with a period $\pi/m_B$.

The angular frequencies of the values of the invariants, $2m_A$ or $2m_B$, are the same as that of the Zitterbewegung \cite{breit} of the respective Dirac particle. For an electron or positron this frequency is approximately $1.55\times10^{21}$ radians per second which may be challenging to observe. But for Dirac quasiparticles it may be less difficult, see e.g. Ref. \cite{katn}.

More generally we can consider the case of constant momenta $p^A_\mu,p^B_\mu$ and constant four-potentials $A^A_\mu,A^B_\mu$ in Alice's and Bob's labs. Then the Hamiltonians $H_A$ and $H_B$ are independent of $t_A$ and $t_B$, respectively, and the evolution is $e^{iH_At_A}\otimes e^{iH_Bt_B}$.  The absolute value of $I_1$ is then periodic in $t_A$ with a period $\pi/\sqrt{\sum_{\mu=1,2,3}(p^A_\mu+A^A_\mu)^2+m_A^2}$ and periodic in $t_B$ with a period $\pi/\sqrt{\sum_{\mu=1,2,3}(p^B_\mu+A^B_\mu)^2+m_B^2}$. Likewise, $I_{2A}$ is periodic in $t_A$ with a period $\pi/\sqrt{\sum_{\mu=1,2,3}(p^A_\mu+A^A_\mu)^2+m_A^2}$ and $I_{2B}$ is periodic in $t_B$ with a period $\pi/\sqrt{\sum_{\mu=1,2,3}(p^B_\mu+A^B_\mu)^2+m_B^2}$.

\section{Reduced one-party matrices}\label{gloin}

From the matrix $\Psi_{AB}$ containing the state coefficients of the shared state of Alice and Bob we can construct reduced matrices associated to either Alice's lab or Bob's lab. Depending on the construction these have different transformation properties under the spinor representations of the local proper orthochronous Lorentz groups.

First we consider the analogues of reduced density matrices from non-relativistic quantum mechanics.
These are constructed as

 \begin{eqnarray}
\rho_A&&\equiv \Psi_{AB} \Psi_{AB}^\dagger\nonumber\\
\rho_B&&\equiv \Psi_{AB}^T \Psi_{AB}^*,
\end{eqnarray}
where $\rho_A$ is the reduced density matrix corresponding to Alice's lab, and $\rho_B$ is the reduced density matrix corresponding to Bob's lab. For transformations $S_A$ in Alice's lab and transformations $S_B$ in Bob's lab these matrices transform as

 \begin{eqnarray}
\Psi_{AB} \Psi_{AB}^\dagger &&\to S_A\Psi_{AB}S_B^T S_B^*\Psi_{AB}^\dagger S_A^\dagger\nonumber\\
 \Psi_{AB}^T \Psi_{AB}^* &&\to S_B\Psi_{AB}^TS_A^T S_A^*\Psi_{AB}^*  S_B^\dagger.
\end{eqnarray}

Since the local unitary evolutions generated by Dirac Hamiltonians and the spinor representations of proper orthochronous Lorentz transformations are full rank, and since for a complex matrix $X$ we have $\textrm{rank}[XX^\dagger]=\textrm{rank}[X]$, it follows that both local unitary evolution and proper orthochronous Lorentz transformations preserve the rank of $\rho_A$ and of $\rho_B$ and $\textrm{rank}[\rho_A]=\textrm{rank}[\rho_B]=\textrm{rank}[\Psi_{AB}]$.
If the shared state of Alice and Bob is a product state $\textrm{rank}[\Psi_{AB}]=1$ and thus the reduced density matrices $\rho_A$ and $\rho_B$ are both rank 1. Such reduced density matrices are called {\it pure}. If on the other hand the shared state is entangled the reduced density matrix corresponding to Alice and the reduced density matrix corresponding to Bob both have rank 2 or greater. A reduced density matrix with rank at least 2 is called {\it mixed}. 
 Moreover, since spinor representations of proper orthochronous Lorentz transformations have determinant 1 and
 the local unitary evolutions generated by Dirac Hamiltonians have determinants with absolute value 1 it follows that both preserve the determinant of the reduced density matrices and $\det[\rho_A]=\det[\rho_B]=\det[\Psi_{AB}]\det[\Psi_{AB}^\dagger]=|I_3|^2$.

The eigenvalues of $\rho_A$ and $\rho_B$ are preserved by local unitary evolution generated by Dirac Hamiltonians and under the spinor representations of rotations, but in general they are not preserved under the spinor representations of Lorentz boosts. Therefore the measures of mixedness used in non-relativistic quantum mechanics that are functions of the eigenvalues such as the {\it purity} and the {\it von Neumann entropy} are in general not invariant. The purity is defined for a reduced density matrix normalized to have trace 1 as $\Tr[\rho^2]$.
A pure reduced density matrix with trace 1 is a projector, i.e., $\rho^2=\rho$ and $\Tr[\rho^2]=1$.  For a mixed reduced density matrix with trace 1 we instead have $\rho^2\neq\rho$ and $\Tr[\rho^2]<1$. 
Thus a purity of less than one indicates entanglement in the shared state of Alice and Bob. The von Neumann entropy is defined as $-\Tr[\rho\ln(\rho)]$ \cite{neumann} and for a normalized reduced density matrix it is zero when the reduced density matrix is pure but nonzero otherwise.

While the purity and von Neumann entropy are in general not invariant under Lorentz transformations the preservation of the rank implies that they are still indicators of entanglement for normalized states. If a reduced density matrix is normalized to have unit trace, a von Neumann entropy greater than zero indicates entanglement of the shared state. Likewise, for a reduced density matrix normalized to have unit trace a purity of less than one indicates entanglement of the shared state.

Beyond the reduced density matrices we can construct other reduced matrices associated with either Alice's lab or Bob's lab. There are such matrices that transform non-trivially only under the spinor representation of the proper orthochronous Lorentz group acting in the lab they are associated with, and that are explicitly invariant under the spinor representation of the proper orthochronous Lorentz group acting in the lab they are not associated with.
Their eigenvalues are invariant under the spinor representations of the proper orthochronous Lorentz groups in both labs.  One such matrix is 

 \begin{eqnarray}
\tilde{\rho}_A\equiv C\Psi_{AB} C\Psi_{AB}^T.
\end{eqnarray}
It is invariant under the spinor representations of proper orthochronous Lorentz transformations of Bob's particle and invariant, up to a U(1) phase, under unitary evolutions generated by zero-mass Dirac Hamiltonians of Bob's particle.
Its eigenvalues are invariant under spinor representations of proper orthochronous Lorentz transformations acting on both particles and invariant, up to a U(1) phase, under unitary evolutions generated by zero-mass Dirac Hamiltonians acting on both particles. In particular $\det[\tilde{\rho}_A]=\det[\Psi_{AB}]^2=I_3^2$ and $\Tr[\tilde{\rho}_A]=2I_1$ and
the eigenvalues are $1/2(I_1-\sqrt{I_1^2-4I_3})$ with multiplicity 2 and $1/2(I_1+\sqrt{I_1^2-4I_3})$ with multiplicity 2.

Similarly, the matrix
 \begin{eqnarray}
\tilde{\rho}_B\equiv C\Psi_{AB}^T C\Psi_{AB},
\end{eqnarray}
is invariant under the spinor representations of proper orthochronous Lorentz transformations of Alice's particle and invariant, up to a U(1) phase, under unitary evolutions generated by zero-mass Dirac Hamiltonians of Alice's particle.
Its eigenvalues are invariant under spinor representations of proper orthochronous Lorentz transformations acting on both particles and invariant, up to a U(1) phase, under unitary evolutions generated by zero-mass Dirac Hamiltonians acting on both particles. In particular $\det[\tilde{\rho}_B]=\det[\Psi_{AB}]^2=I_3^2$ and $\Tr[\tilde{\rho}_B]=2I_1$ and its eigenvalues are the same as those of $\tilde{\rho}_A$, i.e.,
 $1/2(I_1-\sqrt{I_1^2-4I_3})$ with multiplicity 2 and $1/2(I_1+\sqrt{I_1^2-4I_3})$ with multiplicity 2.

Since for square matrices of the same dimension $X,Y$ we have that $\textrm{rank}[XY]\leq \min(\textrm{rank}[X],\textrm{rank}[ Y])$  it follows that
$\textrm{rank}[\tilde{\rho}_A]=\textrm{rank}[\tilde{\rho}_B]\leq \textrm{rank}[ \Psi_{AB}]=\textrm{rank}[\rho_{A}]=\textrm{rank}[\rho_{B}]$. Therefore if $I_3\neq 0$ we have $\textrm{rank}[\tilde{\rho}_A]=\textrm{rank}[\tilde{\rho}_B]=\textrm{rank}[{\rho}_A]=\textrm{rank}[{\rho}_B]=4$ and if $I_3= 0$ but $I_1\neq 0$ we have $2=\textrm{rank}[\tilde{\rho}_A]=\textrm{rank}[\tilde{\rho}_B]\leq\textrm{rank}[{\rho}_A]=\textrm{rank}[{\rho}_B]<4$.

The matrix

 \begin{eqnarray}
\tilde{\eta}_A\equiv C\gamma^5\Psi_{AB} C\gamma^5\Psi_{AB}^T,
\end{eqnarray}
is invariant under the spinor representations of proper orthochronous Lorentz transformations of Bob's particle and invariant, up to a U(1) phase, under unitary evolutions generated by arbitrary-mass Dirac Hamiltonians of Bob's particle.
Its eigenvalues are invariant under spinor representations of proper orthochronous Lorentz transformations acting on both particles and invariant, up to a U(1) phase, under unitary evolutions generated by arbitrary-mass Dirac Hamiltonians acting on both particles. In particular $\det[\tilde{\eta}_A]=\det[\Psi_{AB}]^2=I_3^2$ and $\Tr[\tilde{\eta}_A]=2I_2$ and its eigenvalues  are $1/2(I_2-\sqrt{I_2^2-4I_3})$ with multiplicity 2 and $1/2(I_2+\sqrt{I_2^2-4I_3})$ with multiplicity 2.

Similarly, the matrix
 \begin{eqnarray}
\tilde{\eta}_B\equiv C \gamma^5\Psi_{AB}^T C\gamma^5\Psi_{AB},
\end{eqnarray}
is invariant under the spinor representations of proper orthochronous Lorentz transformations of Alice's particle and invariant, up to a U(1) phase, under unitary evolutions of Alice's particle generated by arbitrary-mass Dirac Hamiltonians.
Its eigenvalues are invariant under spinor representations of proper orthochronous Lorentz transformations acting on both particles and invariant, up to a U(1) phase, under unitary evolutions generated by arbitrary-mass Dirac Hamiltonians acting on both particles. In particular $\det[\tilde{\eta}_B]=\det[\Psi_{AB}]^2=I_3^2$ and $\Tr[\tilde{\eta}_B]=2I_2$ an its eigenvalues are the same as those of $\tilde{\eta}_A$, i.e., $1/2(I_2-\sqrt{I_2^2-4I_3})$ with multiplicity 2 and $1/2(I_2+\sqrt{I_2^2-4I_3})$ with multiplicity 2.

Since $\textrm{rank}[\tilde{\eta}_A]=\textrm{rank}[\tilde{\eta}_B]\leq \textrm{rank}[\Psi_{AB}]=\textrm{rank}[\rho_{A}]=\textrm{rank}[\rho_{B}]$ it follows that if $I_3\neq 0$ we have $\textrm{rank}[\tilde{\eta}_A]=\textrm{rank}[\tilde{\eta}_B]=\textrm{rank}[{\rho}_A]=\textrm{rank}[{\rho}_B]=4$ and if $I_3= 0$ but $I_2\neq 0$ we have $2=\textrm{rank}[\tilde{\eta}_A]=\textrm{rank}[\tilde{\eta}_B]\leq\textrm{rank}[{\rho}_A]=\textrm{rank}[{\rho}_B]<4$.

The matrix

 \begin{eqnarray}
\tilde{\zeta}_A\equiv C\Psi_{AB} C\gamma^5\Psi_{AB}^T,
\end{eqnarray}
is invariant under the spinor representations of proper orthochronous Lorentz transformations of Bob's particle and invariant, up to a U(1) phase, under unitary evolutions of Bob's particle generated by arbitrary-mass Dirac Hamiltonians.
Its eigenvalues are invariant under spinor representations of proper orthochronous Lorentz transformations acting on both particles and invariant, up to a U(1) phase, under unitary evolutions generated by arbitrary-mass Dirac Hamiltonians for Bob's particle and zero-mass Dirac Hamiltonians for Alice's particle. In particular $\det[\tilde{\zeta}_A]=\det[\Psi_{AB}]^2=I_3^2$ and $\Tr[\tilde{\zeta}_A]=2I_{2A}$. 
The eigenvalues are $1/2(I_{2A}-\sqrt{I_{2A}^2+4I_3})$ with multiplicity 2 and $1/2(I_{2A}+\sqrt{I_{2A}^2+4I_3})$ with multiplicity 2.

Similarly, the matrix
 \begin{eqnarray}
\tilde{\zeta}_B\equiv C \gamma^5\Psi_{AB}^T C\Psi_{AB},
\end{eqnarray}
is invariant under the spinor representations of proper orthochronous Lorentz transformations of Alice's particle and invariant, up to a U(1) phase, under unitary evolutions of Alice's particle generated by zero-mass Dirac Hamiltonians.
Its eigenvalues are invariant under spinor representations of proper orthochronous Lorentz transformations acting on both particles and invariant, up to a U(1) phase, under unitary evolutions generated by arbitrary-mass Dirac Hamiltonians for Bob's particle and zero-mass Dirac Hamiltonians for Alice's particle. In particular $\det[\tilde{\zeta}_B]=\det[\Psi_{AB}]^2=I_3^2$ and $\Tr[\tilde{\zeta}_B]=2I_{2A}$. 
The eigenvalues are $1/2(I_{2A}-\sqrt{I_{2A}^2+4I_3})$ with multiplicity 2 and $1/2(I_{2A}+\sqrt{I_{2A}^2+4I_3})$ with multiplicity 2.

Since $\textrm{rank}[\tilde{\zeta}_A]=\textrm{rank}[\tilde{\zeta}_B]\leq \textrm{rank}[\Psi_{AB}]=\textrm{rank}[\rho_{A}]=\textrm{rank}[\rho_{B}]$ it follows that if $I_3\neq 0$ we have $\textrm{rank}[\tilde{\zeta}_A]=\textrm{rank}[\tilde{\zeta}_B]=\textrm{rank}[{\rho}_A]=\textrm{rank}[{\rho}_B]=4$ and if $I_3= 0$ but $I_{2A}\neq 0$ we have $2=\textrm{rank}[\tilde{\zeta}_A]=\textrm{rank}[\tilde{\zeta}_B]\leq\textrm{rank}[{\rho}_A]=\textrm{rank}[{\rho}_B]<4$.

The matrix

 \begin{eqnarray}
\tilde{\kappa}_A\equiv C\gamma^5\Psi_{AB} C\Psi_{AB}^T,
\end{eqnarray}
is invariant under the spinor representations of proper orthochronous Lorentz transformations of Bob's particle and invariant, up to a U(1) phase, under unitary evolutions of Bob's particle generated by zero-mass Dirac Hamiltonians.
Its eigenvalues are invariant under spinor representations of proper orthochronous Lorentz transformations acting on both particles and invariant, up to a U(1) phase, under unitary evolutions generated by arbitrary-mass Dirac Hamiltonians for Alice's particle and zero-mass Dirac Hamiltonians for Bob's particle. In particular $\det[\tilde{\kappa}_A]=\det[\Psi_{AB}]^2=I_3^2$ and $\Tr[\tilde{\kappa}_A]=2I_{2B}$. 
The eigenvalues are $1/2(I_{2B}-\sqrt{I_{2B}^2+4I_3})$ with multiplicity 2 and $1/2(I_{2B}+\sqrt{I_{2B}^2+4I_3})$ with multiplicity 2.

Similarly, the matrix
 \begin{eqnarray}
\tilde{\kappa}_B\equiv C \Psi_{AB}^T C\gamma^5\Psi_{AB},
\end{eqnarray}
is invariant under the spinor representations of proper orthochronous Lorentz transformations of Alice's particle and invariant, up to a U(1) phase, under unitary evolutions of Alice's particle generated by arbitrary-mass Dirac Hamiltonians.
Its eigenvalues are invariant under spinor representations of proper orthochronous Lorentz transformations acting on both particles and invariant, up to a U(1) phase, under unitary evolutions generated by arbitrary-mass Dirac Hamiltonians for Alice's particle and zero-mass Dirac Hamiltonians for Bob's particle. In particular $\det[\tilde{\kappa}_B]=\det[\Psi_{AB}]^2=I_3^2$ and $\Tr[\tilde{\kappa}_B]=2I_{2B}$. 
The eigenvalues are $1/2(I_{2B}-\sqrt{I_{2B}^2+4I_3})$ with multiplicity 2 and $1/2(I_{2B}+\sqrt{I_{2B}^2+4I_3})$ with multiplicity 2.

Since $\textrm{rank}[\tilde{\kappa}_A]=\textrm{rank}[\tilde{\kappa}_B]\leq \textrm{rank}[\Psi_{AB}]=\textrm{rank}[\rho_{A}]=\textrm{rank}[\rho_{B}]$ it follows that if $I_3\neq 0$ we have $\textrm{rank}[\tilde{\kappa}_A]=\textrm{rank}[\tilde{\kappa}_B]=\textrm{rank}[{\rho}_A]=\textrm{rank}[{\rho}_B]=4$ and if $I_3= 0$ but $I_{2A}\neq 0$ we have $2=\textrm{rank}[\tilde{\kappa}_A]=\textrm{rank}[\tilde{\kappa}_B]\leq\textrm{rank}[{\rho}_A]=\textrm{rank}[{\rho}_B]<4$.

\subsection{Mixedness of reduced density matrices in relation to maximal absolute values of the invariants $I_1,I_2,I_{2A},I_{2B}$ and $I_{3}$}

When the absolute value of a determinant $|\det[{X}]|$ is maximized over the complex $d\times d$ matrices $X$ subject to a constraint $\Tr[X X^\dagger]=N$ for some $N\in \mathbb{R}^+$ we have that $X X^\dagger=Nd^{-1} I_d$ where $I_d$ is the $d\times d$ identity matrix. This follows since $X X^\dagger$ is a positive Hermitian matrix and its eigenvalues $z_i$ satisfy $(1/d\sum z_i)^d\geq \prod z_i$, i.e., the arithmetic mean is greater or equal to the geometric mean, with equality if and only if $z_1=z_2=\dots=z_d$. Restated this inequality is $(1/d\Tr[X X^\dagger])^d\geq \det[X X^\dagger]$. Thus if $\Tr[X X^\dagger]=N$ and the maximal value of $\det[X X^\dagger]$ subject to this constraint is achieved it follows that $X X^\dagger=Nd^{-1} I_d$. Since $\det[X X^\dagger]=\det[X]\det[ X^\dagger]=|\det[{X}]|^2$ this implies that the maximal value of $|\det[{X}]|$ is achieved.
Therefore, when $|I_3|$ takes its maximal value for a given normalization of a state the reduced density matrices are proportional to the identity. If we consider reduced density matrices of two Dirac spinors, normalized to have trace one, the minimal purity is achieved for a reduced density matrix proportional to the identity matrix and is $1/4$.  For the same case the von Neumann entropy reaches its maximal value $\ln(4)$. A reduced density matrix for which the purity is minimal or equivalently the von Neumann entropy maximal is called maximally mixed. The corresponding shared entangled state is called {\it maximally entangled}.
Since $I_3$ is a determinant it is invariant under local SL(4,$\mathbb{C}$) transformations. It is a general property of polynomials in the state coefficients invariant under local SL(4,$\mathbb{C}$) transformations that they achieve their maximal absolute value on a set of maximally entangled states as shown in Theorem 1 of Ref. \cite{verstraete}.

The other four Lorentz invariants $I_1,I_2,I_{2A}$ and $I_{2B}$ are not invariant under local SL(4,$\mathbb{C}$) and thus Theorem 1 of Ref. \cite{verstraete} does not apply. Their maximal absolute values for a given normalization are not necessarily reached on a set that consists of only maximally entangled states.
Each of $I_1,I_2,I_{2A}$ and $I_{2B}$ is a sum of four determinants of $2\times 2$ non-overlapping minors of $\Psi_{AB}$ such that each pair of minors form either two full rows, two full columns, or does not form a single full row or column. Let $M_k$ be the minors
corresponding to one of the invariants. Then if the given invariant has achieved its maximal absolute value for a normalized state the minors satisfy $M_kM_k^\dagger=\omega_k I$ where $\omega_k\geq 0$ and $\sum_k{\omega_k}=1/2$. The eigenvalues of the reduced density matrix
are $1/4\pm 1/2\sqrt{1/4-4 (\omega_2 \omega_3 + \omega_1 \omega_4  + x_1)}$ and $1/4\pm 1/2\sqrt{1/4-4 (\omega_2 \omega_3 + \omega_1 \omega_4  + x_2)}$ where $x_1$ and $x_2$ are the eigenvalues of $M_1 M_3^\dagger M_4 M_2^\dagger+M_2 M_4^\dagger M_3 M_1^\dagger$.  Optimizing over the $M_k$ gives that the values of $x_1$ and $x_2$ satisfy $-2\sqrt{\omega_1 \omega_2 \omega_3 \omega_4 }\leq x_1\leq 2\sqrt{\omega_1 \omega_2 \omega_3 \omega_4 }$ and $-2\sqrt{\omega_1 \omega_2 \omega_3 \omega_4 }\leq x_2\leq 2\sqrt{\omega_1 \omega_2 \omega_3 \omega_4 }$ with any values of $x_1$ and $x_2$ in these intervals independently achievable. It follows that $0\leq(\omega_2 \omega_3 + \omega_1 \omega_4  + x_1)\leq 1/16$ and $0\leq(\omega_2 \omega_3 + \omega_1 \omega_4  + x_2)\leq 1/16$ with any values in these intervals of the two expressions independently achievable.
Therefore,  if any of $I_1,I_2,I_{2A}$ and $I_{2B}$ achieves its maximal absolute value the eigenvalues of the reduce density matrices are constrained to the form $1/4-\alpha,1/4+\alpha,1/4+\beta,1/4-\beta $ for $0\leq\alpha\leq 1/4$ and $0\leq\beta\leq 1/4$ but any such $\alpha$ and $\beta$ are allowed. Consequently, for a normalized state any value of the purity between 1/2 and 1/4 or alternatively any value of the von Neumann entropy between $\ln(2)$ and $\ln(4)$ is compatible with a maximal absolute value of any of the Lorentz invariants $I_1,I_2,I_{2A}$ and $I_{2B}$. Thus a maximal absolute value of any of the Lorentz invariants $I_1,I_2,I_{2A}$ and $I_{2B}$ only implies an upper bound on the purity of 1/2 or equivalently a lower bound of $\ln(2)$ of the von Neumann entropy.

\section{The zero locus of all homogeneous polynomial Lorentz invariants}\label{locus}

It is not clear if additional homogeneous polynomials exist, beyond $I_1,I_{2}$, $I_{2A},I_{2B}$ and $I_{3}$ and their algebraic combinations, that are invariant under the proper orthochronous Lorentz group as well as local unitary evolution generated by either zero- or arbitrary mass Dirac Hamiltonians. But even if such additional polynomials exist it is not possible to distinguish all inequivalent forms of entanglement using homogeneous polynomial Lorentz invariants.
The zero locus of the set of all such polynomials, i.e., the set of states for which no such polynomial is non-zero, contains entangled states. The entangled states in the zero locus cannot be distinguished by these polynomials from product states.

One way to see why there is entangled states in the zero locus is to note that there exist parametrized families of spinor representations of local Lorentz boosts for which some entangled states are eigenstates. Theses parametrized families are such that by varying the parameters the eigenvalues multiplying the eigenstates can be any positive real number. No homogeneous polynomial can be invariant under the action of such a family of spinor representations of Lorentz boosts unless it is identically zero on all such eigenstates of the family.

As an example we can consider the family of states 

\begin{eqnarray}
&&\alpha (\phi_0-\phi_2)\otimes(\phi_0-\phi_2)+\beta(\phi_0-\phi_2)\otimes(\phi_1+\phi_3)\nonumber\\
+&&\gamma(\phi_1+\phi_3)\otimes(\phi_0-\phi_2)+\delta (\phi_1+\phi_3)\otimes(\phi_1+\phi_3).
\end{eqnarray}
As long as $\alpha \delta-\beta \gamma \neq 0$ the corresponding state is entangled and $\Psi_{AB}$ has rank 2. Otherwise the state is a product state and $\Psi_{AB}$ has rank 1.

This family of states are eigenstates of the family $\exp\left(a \gamma^0\gamma^3\right)\otimes\exp\left(b \gamma^0\gamma^3\right)$ of spinor representations of Lorentz boosts parametrized by $a,b\in \mathbb{R}$ and where $\exp\left(a \gamma^0\gamma^3\right)$ acts on Alice's particle and $\exp\left(b \gamma^0\gamma^3\right)$ acts on Bob's particle. The eigenvalues are $e^{-(a+b)}$ and thus the eigenvalue can be any positive real number for some choice of $a$ and $b$.

Note that this reason why the zero locus of all the homogeneous Lorentz invariant polynomials contain entangled states is analogous to the reason why for a system of three or more non-relativistic two-component spinors the zero locus of all polynomials invariant under local $\mathrm{SL}(2,\mathbb{C})$ operations contain entangled states \cite{dur}. An example are the so called W-states defined for three or more non-relativistic two-component spinors \cite{dur}.

\section{The local reversible operations and the invariance groups of $|I_1|,|I_2|,|I_{2A}|,|I_{2B}|$ and $|I_3|$}\label{lie}
In Section \ref{spen} the local reversible operations were defined as the local operations that can be implemented by local unitary evolution of the system and changes of local reference frames. The unitary evolutions were defined in Section \ref{ham} as generated by time dependent, bounded and strongly continuous Hamiltonians. The set of local reversible transformations that can be implemented on a pair of Dirac spinors thus depends on the set of physically allowed Hamiltonians. 
A few different classes of Hamiltonians defined in terms of different degrees in the gamma matrices were considered in Section \ref{ham}.

In this Appendix we elaborate on the properties of the sets of reversible operations, including properties that depend on the specific class of allowed Hamiltonians.
We also consider the different groups of transformations that leave the bilinear forms $\psi^T C\varphi$ and $\psi^T C\gamma^5\varphi$ and the polynomials $I_1,I_2,I_{2A},I_{2B}$ and $I_3$ invariant up to a U(1) phase. The relations between these groups and the sets of reversible operations are described.
For this purpose we use the following Theorem

\begin{theorem}\label{tbv}
Let $G$ be a continuous connected group generated by the elements $\exp (X_i t)$ for $t\in \mathbb{R}$ and $X_i\in S$ where $S$ is a set of matrices. 
Then $\exp ((a X_i+b X_j)t)$ and $\exp ([X_i,X_j]t)$ are in the closure $\mathrm{cl}\phantom{i} G$ of $G$ for all $a,b,t\in \mathbb{R}$ and $X_i,X_j\in S$. Let $\mathfrak{h}$ be the real Lie algebra spanned by the $X_i\in S$ and the commutators $[X_i,X_j]$ for $X_i,X_j\in S$.
Let $H$ be the connected matrix Lie group that is generated by the exponentials $\exp (Z_i )$ for $Z_i\in \mathfrak{h}$.
Then $G$ is a dense subset of $H$, i.e., for any $h\in H$ every neighbourhood of $h$ contains at least one element of $G$. If $G$ contains its limit points we have that $G=H$.
\end{theorem}
\begin{proof}
Let $X,Y$ be such that $\exp (Xt)\in G$ and $\exp (Yt)\in G$ for all $t\in \mathbb{R}$.
Consider the Baker-Campbell-Hausdorff formula
 \begin{eqnarray}
\exp (aXt)\exp (bYt)=\exp\left(aXt+bYt+[aX,bY]\frac{t^2}{2}+\mathcal{O}(t^3)\right),\nonumber\\
\end{eqnarray}
where the third and higher order terms in $t$ are denoted $\mathcal{O}(t^3)$.  
From this formula follows that
\begin{eqnarray}\label{lab1}
\left(\exp \left(\frac{ta}{n}X\right)\exp \left(\frac{tb}{n}Y\right)\right)^n=\exp\left((aX+bY)t+\mathcal{O}\left(\frac{1}{n}\right)\right),
\end{eqnarray}
where first and higher order terms in $\frac{1}{n}$ are denoted $\mathcal{O}(\frac{1}{n})$.
Taking the limit $n\to\infty$ gives $\exp ((a X+b Y)t)$. 
Thus a sequence of elements in $G$ exists that converges to $\exp ((a X+b Y)t)$. By definition $\exp ((a X+b Y)t)$ is a limit point of $G$ and belongs to $\mathrm{cl}\phantom{i} G$ for all $t\in \mathbb{R}$.
Next consider
\begin{eqnarray}\label{lab2}
&&\left(\exp \left(\frac{t}{n}X\right)\exp \left(\frac{t}{n}Y\right)\exp \left(-\frac{t}{n}X\right)\exp \left(-\frac{t}{n}Y\right)\right)^{n^2}\nonumber\\
=&&\exp\left([X,Y]t^2+\mathcal{O}\left(\frac{1}{n}\right)\right).
\end{eqnarray}
Taking the limit $n\to\infty$ gives  $\exp ([X,Y]t^2)$. Exchanging the roles of $X$ and $Y$ gives the limit $\exp ([Y,X]t^2)=\exp (-[X,Y]t^2)$. Thus $\exp ([X,Y]t)$ is a limit point of $G$ and belongs to $\mathrm{cl}\phantom{i} G$ for all $t\in \mathbb{R}$.
Let $Z$ be the smallest real matrix Lie algebra that contains all $X,Y$ such that $\exp (Xt)\in G$ and $\exp (Yt)\in G$ for all $t\in \mathbb{R}$.
By concatenated use of Eq. (\ref{lab1}) and Eq. (\ref{lab2}) we see that every $\exp (Z)$ for $Z\in{\mathfrak{h}}$ is in $\mathrm{cl}\phantom{i} G$.  
Moreover if $H$ is the connected matrix Lie group generated by the $\exp (Z)$ for $Z\in{\mathfrak{h}}$, by definition every $h\in H$ can be expressed as $h=\exp (Z_1)\exp (Z_2)\dots \exp (Z_m)$ for $Z_j\in{\mathfrak{h}}$.  Therefore every $h\in H$ is in $\mathrm{cl}\phantom{i} G$. Since $G$ is a subgroup of $H$ we have that $\mathrm{cl}\phantom{i} G=H$.
See also e.g. Ref. \cite{wall} Ch. 1.3.3. or Ref. \cite{hall}  Ch. 3.3.
\end{proof}

Based on Theorem \ref{tbv} we can make the following Corollary

\begin{corollary}\label{tbc}
Let $G$ be a continuous connected group generated by the elements $\exp (X_i t)$ for $t\in \mathbb{R}$ and $X_i\in S$ where $S$ is a set of matrices. Let $\mathfrak{h}$ be the real Lie algebra spanned by the $X_i\in S$ and the commutators $[X_i,X_j]$ for $X_i,X_j\in S$.
Let $H$ be the connected matrix Lie group that is generated by the exponentials $\exp (Z_i )$ for $Z_i\in \mathfrak{h}$.
Let $F$ be a subset of $H$ such that for any $g\in G$ every neighbourhood of $g$ contains an element in $F$ and let $F$ be a semigroup. 
Then $F$ is a dense subset of $H$, i.e., for any $h\in H$ every neighbourhood of $h$ contains at least one element of $F$. If $F$ contains its limit points we have that $F=H$.
\end{corollary}
\begin{proof}
Consider $\exp (Xt)$ for $X\in S$, $t\in \mathbb{R}$ and a neighbourhood $U_\epsilon$ of $\exp (Xt)$ defined by $U_\epsilon\equiv\{\exp (Xt)+W\epsilon |\phantom{u}||W||_2\leq 1\}$ for some $\epsilon>0$, where $||\cdot||_2$ is the spectral norm. 
By assumption for every $X\in S$, $t\in \mathbb{R}$, and $\epsilon>0$ there exists an element $f\in F$ such that $f\in U_\epsilon$. Since the spectral norm is submultiplicative we have that $||WM||_2\leq ||M||_2$ for every matrix $M$ and every $W$ such that $\exp (Xt)+W\epsilon\in U_\epsilon$. Therefore for any product $\epsilon^nM_0W_1M_1W_2M_2\dots W_nM_n$ where $||W_k||_2\leq 1$ for all $k$ we have $||\epsilon^nM_0W_1M_1W_2M_2\dots W_nM_n||_2\leq \epsilon^n||M_0||_2||M_1||_2\dots||M_n||_2$. Let the $M_k$ be matrices with finite spectral norm that do not depend on $\epsilon$ and consider sequences $(\epsilon_j, W_{jk})$ such that $\epsilon_j>\epsilon_{j+1}>0$ and $||W_{jk}||_2\leq 1$ with $\lim_{j\to \infty}\epsilon_j= 0$.
Then we have that $\lim_{j\to \infty}||\epsilon_j^nM_0W_{j1}M_1W_{j2}M_2\dots W_{jn}M_n||_2\leq\lim_{j\to \infty}\epsilon_j^n||M_0||_2||M_1||_2\dots||M_n||_2=0$ and thus $\lim_{j\to \infty}\epsilon_j^nM_0W_{j1}M_1W_{j2}M_2\dots W_{jn}M_n=0$ for any positive integer $n$. 

Denote by $\mathcal{O}\left(\epsilon\right)$ any sum on the form $\sum_k\sum_{n>0} \epsilon_j^nM_{0,k}W_{j1,k}M_{1,k}W_{j2,k}M_{2,k}\dots W_{jn,k}M_{n,k}$ where the $M_{l,k}$ have finite spectral norm and do not depend on $\epsilon$ and $||W_{jl,k}||_2\leq 1$ for all $jl,k$. Then for any $X,Y\in S$ and any $\epsilon_j>0$ by assumption there exist elements $f_1,f_2\in F$ such that $f_1f_2=\exp \left(\frac{ta}{n}X\right)\exp \left(\frac{tb}{n}Y\right)+\mathcal{O}\left(\epsilon_j\right)$. 
 By the Baker-Campbell-Hausdorff formula we have 
\begin{eqnarray}\label{lab1x}
&&\left(\exp \left(\frac{ta}{n}X\right)\exp \left(\frac{tb}{n}Y\right)+\mathcal{O}\left(\epsilon_j\right)\right)^n\nonumber\\=&&\exp\left((aX+bY)t+\mathcal{O}\left(\frac{1}{n}\right)\right)+\mathcal{O}\left(\epsilon_j\right),
\end{eqnarray}
where first and higher order terms in $\frac{1}{n}$ are denoted $\mathcal{O}(\frac{1}{n})$. 
Taking the limit $n\to\infty$ and $\epsilon_j\to 0$ gives $\exp ((a X+b Y)t)$.  Since $F$ is a semigroup $(f_1f_2)^n\in F$ for any $n>0$ and thus $\exp ((a X+b Y)t)$ is a limit point of $F$ for all $t\in \mathbb{R}$.
Similarly, for any $X,Y\in S$ and any $\epsilon_j>0$ by assumption there exist elements $f_1,f_2,f_3,f_4\in F$ such that $f_1f_2f_3f_4=\exp \left(\frac{t}{n}X\right)\exp \left(\frac{t}{n}Y\right)\exp \left(-\frac{t}{n}X\right)\exp \left(-\frac{t}{n}Y\right)+\mathcal{O}\left(\epsilon_j\right)$. By the Baker-Campbell-Hausdorff formula we have
\begin{eqnarray}\label{lab2x}
&&\left(\exp \left(\frac{t}{n}X\right)\exp \left(\frac{t}{n}Y\right)\exp \left(-\frac{t}{n}X\right)\exp \left(-\frac{t}{n}Y\right)+\mathcal{O}\left(\epsilon_j\right)\right)^{n^2}\nonumber\\
=&&\exp\left([X,Y]t^2+\mathcal{O}\left(\frac{1}{n}\right)\right)+\mathcal{O}\left(\epsilon_j\right).
\end{eqnarray}
Taking the limit $n\to\infty$ and $\epsilon_j\to 0$ gives $\exp ([X,Y]t^2)$. Exchanging the roles of $X$ and $Y$ gives the limit $\exp ([Y,X]t^2)=\exp (-[X,Y]t^2)$. Since $F$ is a semigroup $(f_1f_2f_3f_4)^n\in F$ for any $n>0$ and thus $\exp ([X,Y]t)$ is a limit point of $F$ for all $t\in \mathbb{R}$.
Let $Z$ be the smallest real matrix Lie algebra that contains all $X,Y$ such that $\exp (Xt)\in G$ and $\exp (Yt)\in G$ for all $t\in \mathbb{R}$.
By concatenated use of Eq. (\ref{lab1x}) and Eq. (\ref{lab2x}) we see that every $\exp (Z)$ for $Z\in{\mathfrak{h}}$ is in $\mathrm{cl}\phantom{i} F$.  
Moreover if $H$ is the connected matrix Lie group generated by the $\exp (Z)$ for $Z\in{\mathfrak{h}}$, by definition every $h\in H$ can be expressed as $h=\exp (Z_1)\exp (Z_2)\dots \exp (Z_m)$ for $Z_j\in{\mathfrak{h}}$.  Therefore every $h\in H$ is in $\mathrm{cl}\phantom{i} F$. Since $F$ is a subset of $H$ we have that $\mathrm{cl}\phantom{i} F=H$.
\end{proof}

\subsection{The set of local unitary operations}\label{uni}

Let $S_H$ be a convex set of time independent bounded Hamiltonians and let $S_{H(t)}$ be the set of time dependent strongly continuous Hamiltonians such that 
$H(t)\in S_{H(t)} $ if and only if $H(t)|_{t=s}\in S_{H}$ for all $s$. We then consider the set of unitary transformations that can be implemented by evolutions generated by $H(t)\in S_{H(t)} $.
To do so we first look at the unitary transformations defined as products of exponentials $\exp (iH_j \Delta_j)$ for $H_j\in S_H$ 

\begin{eqnarray}\label{qwes}
e^{iH_n\Delta_n }e^{iH_{n-1}\Delta_{n-1} }\dots e^{iH_1\Delta_1 }e^{iH_{0}\Delta_0 },
\end{eqnarray}
where $\Delta_j\in\mathbb{R}$. If for any $H_j,H_k\in S_H$ we have that $a H_j+b H_k\in S_H$ for all $a,b\in \mathbb{R}$ and $i[H_j,H_k]\in S_H$ the set of skew-Hermitian matrices $iH_j$ for $H_j\in S_H$ is a real matrix Lie algebra and the set of all unitary transformations on the form in Eq. (\ref{qwes}) is a connected matrix Lie group. Even if the set of skew-Hermitian matrices $iH_j$ for $H_j\in S_H$ is not a matrix Lie algebra we can consider the Lie algebra $\overline{\mathfrak{s}}_H$ obtained as the $\mathbb{R}$-algebra over $iH_j$ and $[iH_j,iH_k]$ for $H_j,H_k\in S_H$. Let $G^{\overline{\mathfrak{s}}_H}$ be the connected matrix Lie group that is generated by the exponentials $\exp (zt)$ for $t\in \mathbb{R}$ and $z\in \overline{\mathfrak{s}}_H$.
By Theorem \ref{tbv} the set of all unitary transformations on the form in Eq. (\ref{qwes}) form a dense subset of $G^{\overline{\mathfrak{s}}_H}$.

If the zero matrix is in the relative interior of $S_H$, i.e., if the intersection of some neighbourhood of the zero matrix with the affine hull of $S_H$ is contained in $S_H$, the unitary transformation $\exp (iH_j \Delta_j)$ for $H_j\in S_H$ can be exactly implemented by an evolution generated by a Hamiltonian in $S_{H(t)}$ for any $\Delta_j\in\mathbb{R}$. For any $\Delta_j$ there is some sufficiently large time interval $[0,t]$ for which we can choose a time dependent Hamiltonian $H(s)=H_j f(s)$ where $f(s)$ is bounded, continuous and such that $f(0)=f(t)=0$, $\int_{0}^tf(s)ds=\Delta_j$, and $H_j f(s)\in S_H$ for all $s\in [0,t]$. Thus every unitary on the form in Eq. (\ref{qwes}) can be implemented by evolutions generated by $H(t)\in S_{H(t)} $. Furthermore, since any unitary transformation that can be implemented by an evolution generated by $H(t)\in S_{H(t)}$ is in $G^{\overline{\mathfrak{s}}_H}$ it follows that these unitary transformations is a dense subset of $G^{\overline{\mathfrak{s}}_H}$.

If the zero matrix is not in the relative interior of $S_H$ but the transformations $\exp (iH_j \Delta_j)$ for $\Delta_j\in\mathbb{R}$ is a compact one-parameter group for every $H_j\in S_H$ we can still implement an arbitrarily good approximation of the unitary transformations in Eq. (\ref{qwes}) using Hamiltonians in $S_{H(t)}$ for any $\Delta_j\in\mathbb{R}$. Compactness implies that there exists a $P>0$ such that $\exp (iH_j P)=I$ and thus it is sufficient to consider $P\geq\Delta_j\geq 0$.
 Since $S_H$ is a convex set each factor $\exp (iH_j \Delta_j)$ where $\Delta_j>0$ can be approximated arbitrarily well for a time interval $[0,\Delta_j]$ by $\exp (iH(t))$ where $H(t)=H_{j-1}+(H_j-H_{j-1})\frac{t}{\epsilon}$ for $0\leq t\leq\epsilon$, $H(t)=H_{j+1}+(H_{j}-H_{j+1})\frac{\Delta_j-t}{\epsilon}$ for  $\Delta_j-\epsilon\leq t\leq\Delta_j$ and $H(t)=H_j$ for $\epsilon\leq t\leq \Delta_j-\epsilon$ for arbitrarily small $\epsilon>0$. This construction provides a continuous family of transformations parametrized by $\epsilon$ with $\exp (iH_j \Delta_j)$ as the limit point when $\epsilon\to 0$. 
Moreover, since any unitary transformation that can be implemented by an evolution generated by $H(t)\in S_{H(t)}$  is in $G^{\overline{\mathfrak{s}}_H}$ it follows from Theorem \ref{tbv} and Corollary \ref{tbc} that for any unitary transformation $g\in G^{\overline{\mathfrak{s}}_H}$ every neighbourhood of $g$ contains a unitary operation that can be generated by $H(t)\in S_{H(t)}$.

We can conclude from the above that the set of unitary transformations that can be implemented by evolutions generated by $H(t)\in S_{H(t)} $ is a dense subset of $G^{\overline{\mathfrak{s}}_H}$ if either the zero matrix is in the relative interior of $S_H$ or if the transformations $\exp (iH \Delta)$ for $\Delta\in\mathbb{R}$ is a compact one-parameter group for every $H\in S_H$.

Now consider a convex set of bounded Hamiltonians $S_H$ such that the set of all unitary transformations that can be implemented by strongly continuous Hamiltonians in $S_{H(t)}$ leave the bilinear form $\psi^T C\varphi$ invariant, up to a U(1) phase. 
Then, since $\psi^T C\varphi$ is a continuous function on the Hilbert space it is invariant, up to a U(1) phase, also under all transformations in the Lie group $G^{\overline{\mathfrak{s}}_H}$. By an analogous argument there is a Lie group preserving the bilinear form $\psi^T C\gamma^5\varphi$, up to a U(1) phase, for every set of unitary transformations implemented by strongly continuous Hamiltonians, that preserve $\psi^T C\gamma^5\varphi$ up to a U(1) phase.
We therefore construct these Lie groups for three specific cases of Hamiltonians; The Dirac Hamiltonians with non-zero mass, the Dirac Hamiltonian with zero mass and the Dirac Hamiltonian with zero mass and a coupling to a Yukawa pseudo-scalar boson.

Consider a Dirac particle with momentum $\bold{k}$, non-zero mass $m$ and allow any bounded time dependent values of the electromagnetic four-potential $A_{\mu}$.
The Dirac Hamiltonian for this case is on the form
\begin{eqnarray}\label{tred}
H_D=-\sum_{\mu=1,2,3}\gamma^0\gamma^\mu(k_\mu-qA_{\mu}(t))+qA_0(t)I +m\gamma^0.
\end{eqnarray}
For any time independent Hamiltonian on the form $H_D=\sum_{\mu=1,2,3}a_\mu\gamma^0\gamma^\mu +b\gamma^0$ where $a_\mu,b\in \mathbb{R}$ the elements $\exp (iH_Dt)$ for $t\in\mathbb{R}$ form a compact one-parameter group. 
 Thus from Theorem \ref{tbv} and Corollary \ref{tbc} follows that the set of unitary transformations generated by strongly continuous Dirac Hamiltonians of the kind in Eq. (\ref{tred}) is dense in the connected matrix Lie group $G_U^{C\gamma^5}$ generated by the exponentials of the real Lie algebra spanned by the 11 linearly independent matrices $i\gamma^0$, $\gamma^1$, $\gamma^2$, $\gamma^3$,  $i\gamma^0 \gamma^1$, $i\gamma^0 \gamma^2$, $i\gamma^0 \gamma^3$, $\gamma^1 \gamma^2$, $\gamma^1 \gamma^3$, $\gamma^2 \gamma^3$ and $iI$. If we allow arbitrary Hamiltonians of the type $H^{1,2}(t)+H^0(t)$ described in Eq. (\ref{w1}) any unitary transformation in $G_U^{C\gamma^5}$ can be exactly implemented.
The matrix $iI$ generates the group U(1) and the other matrices generate the group of unitary transformations that preserve the bilinear form $\psi^T C\gamma^5\varphi$ which is isomorphic to
$\mathrm{Sp}(2)$ the compact symplectic group of $4\times 4$ matrices (See e.g. Ref. \cite{hall}  Ch. 1.2.8).
Thus $G_U^{C\gamma^5}$ is isomorphic to $\mathrm{U(1)}\times\mathrm{Sp}(2)$.

Next, consider a Dirac particle with momentum $\bold{k}$ and zero mass with a coupling to a Yukawa pseudo-scalar boson $gi\gamma^0\gamma^5\phi$, and allow any bounded values of the electromagnetic four-potential $A_{\mu}(t)$. 
The Dirac Hamiltonian for this case is on the form
\begin{eqnarray}\label{tred1}
H_D=-\sum_{\mu=1,2,3}\gamma^0\gamma^\mu(k_\mu-qA_{\mu}(t))+qA_0(t)I +gi\gamma^0\gamma^5\phi.
\end{eqnarray}
For any time independent Hamiltonian on the form $H_D=\sum_{\mu=1,2,3}a_\mu\gamma^0\gamma^\mu +bi\gamma^0\gamma^5$ where $a_\mu,b\in \mathbb{R}$ the elements $\exp (iH_Dt)$ for $t\in\mathbb{R}$ form a compact one-parameter group. 
 Thus from Theorem \ref{tbv} and Corollary \ref{tbc} follows that the set of unitary transformations generated by strongly continuous Dirac Hamiltonians of the kind in Eq. (\ref{tred1}) is dense in the connected matrix Lie group $G_U^{C}$ generated by the exponentials of the real Lie algebra spanned by the 11 linearly independent matrices $\gamma^0\gamma^5$, $i\gamma^1\gamma^5$, $i\gamma^2\gamma^5$, $i\gamma^3\gamma^5$, $i\gamma^0 \gamma^1$, $i\gamma^0 \gamma^2$, $i\gamma^0 \gamma^3$, $\gamma^1 \gamma^2$, $\gamma^1 \gamma^3$, $\gamma^2 \gamma^3$ and $iI$. If we allow arbitrary Hamiltonians of the type $H^{2,3}(t)+H^0(t)$ described in Eq. (\ref{w2}) any unitary transformation in $G_U^{C}$ can be exactly implemented. The matrix $iI$ generates the group U(1) and the other matrices generate the group of unitary transformations that preserve the bilinear form $\psi^T C\varphi$ which is isomorphic to
$\mathrm{Sp}(2)$. Thus $G_U^{C}$ is isomorphic to $\mathrm{U(1)}\times\mathrm{Sp}(2)$.

Lastly, consider a zero mass Dirac particle with momentum $\bold{k}$ and allow any bounded values of the electromagnetic four-potential $A_{\mu}(t)$.
The Dirac Hamiltonian for this case is on the form
\begin{eqnarray}\label{tred2}
H_D=-\sum_{\mu=1,2,3}\gamma^0\gamma^\mu(k_\mu-qA_{\mu}(t))+qA_0(t)I.
\end{eqnarray}
For any time independent Hamiltonian on the form $H_D=\sum_{\mu=1,2,3}a_\mu\gamma^0\gamma^\mu$ where $a_\mu\in \mathbb{R}$ the elements $\exp (iH_Dt)$ for $t\in\mathbb{R}$ form a compact one-parameter group. 
 Thus from Theorem \ref{tbv} and Corollary \ref{tbc} follows that the set of unitary transformations generated by strongly continuous Dirac Hamiltonians of the kind in Eq. (\ref{tred2}) is dense in the connected matrix Lie group generated by the exponentials of the real Lie algebra spanned by the 7 linearly independent matrices $i\gamma^0 \gamma^1$, $i\gamma^0 \gamma^2$, $i\gamma^0 \gamma^3$, $\gamma^1 \gamma^2$, $\gamma^1 \gamma^3$, $\gamma^2 \gamma^3$ and $iI$.
 If we allow arbitrary Hamiltonians with only second and zeroth degree terms in the gamma matrices any unitary transformation in  this Lie group can be exactly implemented. The matrix $iI$ generates the group U(1) and the other matrices generate the group of unitary transformations that preserve both the bilinear form $\psi^T C\varphi$ and the bilinear form $\psi^T C\gamma^5\varphi$ which is isomorphic to $\mathrm{SU}(2)\times\mathrm{SU}(2)$. Thus the Lie group is $G_U^{C}\cap G_U^{C\gamma^5} $ and is isomorphic to $\mathrm{U}(1)\times\mathrm{SU}(2)\times\mathrm{SU}(2)$.

\subsection{The invariance groups of $|I_1|,|I_2|,|I_{2A}|,|I_{2B}|$ and $|I_3|$}
In Section \ref{uni} the connected matrix Lie group $G_U^{C}$ of all unitary operations that preserve $\psi^T C\varphi$ up to a U(1) phase was described. Since $\psi^T C\varphi$ is a continuous function and invariant also under the spinor representation of the proper orthochronous Lorentz group it follows from Theorem \ref{tbv} that it is invariant,  up to a U(1) phase, under the smallest connected matrix Lie group that contains $G_U^{C}$ and the spinor representation of the proper orthochronous Lorentz group as subgroups.
Similarly, $\psi^T C\gamma^5\varphi$ is invariant, up to a U(1) phase, under the smallest connected matrix Lie group that contains $G_U^{C\gamma^5}$ and the spinor representation of the proper orthochronous Lorentz group as subgroups. Here we describe these Lie groups.

The generators of the spinor representation of the proper orthochronous Lorentz group from Eq. (\ref{gene}) are given by $S^{\rho\sigma}=\frac{1}{4}[\gamma^\rho,\gamma^\sigma]$. Thus, the Lie algebra of the spinor representation of the proper orthochronous Lorentz group is the $\mathbb{R}$-algebra over the matrices  $\gamma^0 \gamma^1$, $\gamma^0 \gamma^2$, $\gamma^0 \gamma^3$, $\gamma^1 \gamma^2$, $\gamma^1 \gamma^3$, and $\gamma^2 \gamma^3$.

The Lie algebra of $G_U^{C}$ is the $\mathbb{R}$-algebra over the matrices $\gamma^5\gamma^0$, $i\gamma^5\gamma^1$,$i\gamma^5\gamma^2$, $i\gamma^5\gamma^3$, $i\gamma^0 \gamma^1$, $i\gamma^0 \gamma^2$, $i\gamma^0 \gamma^3$, $\gamma^1 \gamma^2$, $\gamma^1 \gamma^3$, $\gamma^2 \gamma^3$ and $iI$.
The smallest Lie algebra that contains both this algebra and the Lie algebra of the spinor representation of the proper orthochronous Lorentz group is the $\mathbb{R}$-algebra over the 21 matrices
$\gamma^5\gamma^0$, $i\gamma^5\gamma^0$, $\gamma^5\gamma^1$, $\gamma^5\gamma^2$, $\gamma^5\gamma^3$, $i\gamma^5\gamma^1$, $i\gamma^5\gamma^2$, $i\gamma^5\gamma^3$, $\gamma^0 \gamma^1$, $\gamma^0 \gamma^2$, $\gamma^0 \gamma^3$, $i\gamma^0 \gamma^1$, $i\gamma^0 \gamma^2$, $i\gamma^0 \gamma^3$, $\gamma^1 \gamma^2$, $\gamma^1 \gamma^3$, $\gamma^2 \gamma^3$, $i\gamma^1 \gamma^2$, $i\gamma^1 \gamma^3$, $i\gamma^2 \gamma^3$, and $iI$.  
The matrix $iI$ generates the group U(1) and the remainder of these matrices generate the group of linear transformations that preserve the bilinear form $\psi^T C\varphi$ which is isomorphic to
$\mathrm{Sp}(4,\mathbb{C})$ the symplectic group of $4\times 4$ matrices (See e.g. Ref. \cite{hall}  Ch. 1.2.4). Thus a continuous function on the Hilbert space that it is invariant, up to a U(1) phase, under $G_U^{C}$ and under spinor representations of local proper orthochronous Lorentz transformations is invariant, up to a U(1) phase, also under all transformations in the connected matrix Lie group $G^{C}$ isomorphic to $\mathrm{U(1)}\times\mathrm{Sp}(4,\mathbb{C})$.

The Lie algebra of $G_U^{C\gamma^5}$ is the $\mathbb{R}$-algebra over the matrices $i\gamma^0$, $\gamma^1$,$\gamma^2$,$\gamma^3$, $i\gamma^0 \gamma^1$, $i\gamma^0 \gamma^2$, $i\gamma^0 \gamma^3$, $\gamma^1 \gamma^2$, $\gamma^1 \gamma^3$, $\gamma^2 \gamma^3$ and $iI$.
The smallest Lie algebra that contains both this algebra and the Lie algebra of the spinor representation of the proper orthochronous Lorentz group is the $\mathbb{R}$-algebra over the 21 matrices
$\gamma^0$, $i\gamma^0$, $\gamma^1$, $\gamma^2$, $\gamma^3$, $i\gamma^1$, $i\gamma^2$, $i\gamma^3$, $\gamma^0 \gamma^1$, $\gamma^0 \gamma^2$, $\gamma^0 \gamma^3$, $i\gamma^0 \gamma^1$, $i\gamma^0 \gamma^2$, $i\gamma^0 \gamma^3$, $\gamma^1 \gamma^2$, $\gamma^1 \gamma^3$, $\gamma^2 \gamma^3$, $i\gamma^1 \gamma^2$, $i\gamma^1 \gamma^3$, $i\gamma^2 \gamma^3$, and $iI$.  
The matrix $iI$ generates the group U(1) and the remainder of these matrices generate the group of linear transformations that preserve the bilinear form $\psi^T C\gamma^5\varphi$ which is isomorphic to
$\mathrm{Sp}(4,\mathbb{C})$. Thus a continuous function on the Hilbert space that it is invariant, up to a U(1) phase, under $G_U^{C\gamma^5}$ and under spinor representations of local proper orthochronous Lorentz transformations is invariant, up to a U(1) phase, also under all transformations in the connected matrix Lie group $G^{C\gamma^5}$ isomorphic to $\mathrm{U(1)}\times\mathrm{Sp}(4,\mathbb{C})$.

We can consider also the Lie algebra of $G_U^{C}\cap G_U^{C\gamma^5}$ which is the $\mathbb{R}$-algebra over the matrices $i\gamma^0 \gamma^1$, $i\gamma^0 \gamma^2$, $i\gamma^0 \gamma^3$, $\gamma^1 \gamma^2$, $\gamma^1 \gamma^3$, $\gamma^2 \gamma^3$ and $iI$.
The smallest Lie algebra that contains both this algebra and the Lie algebra of the spinor representation of the proper orthochronous Lorentz group is the $\mathbb{R}$-algebra over the 13 matrices
$\gamma^0 \gamma^1$, $\gamma^0 \gamma^2$, $\gamma^0 \gamma^3$, $i\gamma^0 \gamma^1$, $i\gamma^0 \gamma^2$, $i\gamma^0 \gamma^3$, $\gamma^1 \gamma^2$, $\gamma^1 \gamma^3$, $\gamma^2 \gamma^3$, $i\gamma^1 \gamma^2$, $i\gamma^1 \gamma^3$, $i\gamma^2 \gamma^3$, and $iI$.  
The matrix $iI$ generates the group U(1) and the remainder of these matrices generate the group of linear transformations that preserve both the bilinear form $\psi^T C\gamma^5\varphi$ and the bilinear form $\psi^T C\varphi$ which is isomorphic to
$\mathrm{SL}(2,\mathbb{C})\times \mathrm{SL}(2,\mathbb{C})$. Thus a continuous function on the Hilbert space that it is invariant, up to a U(1) phase, under $G_U^{C}\cap G_U^{C\gamma^5}$ and under spinor representations of local proper orthochronous Lorentz transformations is invariant, up to a U(1) phase, also under all transformations in the connected matrix Lie group $G^{C}\cap G^{C\gamma^5}$ isomorphic to $\mathrm{U}(1)\times\mathrm{SL}(2,\mathbb{C})\times \mathrm{SL}(2,\mathbb{C})$.

In conclusion, continuous functions on the Hilbert space of two Dirac spinors that are constructed to be invariant under local reversible operations as defined in Sections \ref{rep}, \ref{ham} and \ref{spen} are automatically invariant under a Lie group $G_1\otimes G_2$ where the local groups $G_1$ and $G_2$ depend on the classes of physically allowed local Hamiltonians.
For the case of a local Dirac Hamiltonians with nonzero mass, the local group is $G^{C\gamma^5}$. For a local Dirac Hamiltonian with zero mass and a coupling to a Yukawa pseudo-scalar boson, the local group is $G^{C}$. For massless local Dirac Hamiltonians without additional couplings the local group is $G^{C}\cap G^{C\gamma^5}$.
 The function $|I_1|$ is invariant under $G^{C}\otimes G^{C}$, while  $|I_2|$ is invariant under $G^{C\gamma^5}\otimes G^{C\gamma^5}$, $|I_{2A}|$ is invariant under $G^{C}\otimes G^{C\gamma^5}$ and $|I_{2B}|$ is invariant under $G^{C\gamma^5}\otimes G^{C}$. The function $|I_3|$ is invariant under $\mathrm{U}(1)\times\mathrm{SL}(4,\mathbb{C})\otimes\mathrm{U}(1)\times \mathrm{SL}(4,\mathbb{C})$ and thus invariant under $G^{C}\otimes G^{C}$, $G^{C\gamma^5}\otimes G^{C\gamma^5}$, $G^{C}\otimes G^{C\gamma^5}$ and $G^{C\gamma^5}\otimes G^{C}$. It is the orbits of these respective groups that can be partially distinguished by the Lorentz invariant polynomials.

\section{Lorentz invariance of the convex roof extensions of $|I_1|,|I_2|,|I_{2A}|,|I_{2B}|$ and $|I_3|$}\label{con}

The Lorentz invariants $|I_1|,|I_2|,|I_{2A}|,|I_{2B}|$ and $|I_3|$ are defined for sums of tensor products of Dirac spinors. Every such sum $\xi\equiv\sum_{ij}\psi_{ij}\phi_i\otimes\phi_j$ however can be mapped to a positive semi-definite Hermitian rank 1 matrix $\xi\xi^\dagger$ which up to a positive multiplicative constant is a projector onto the subspace spanned by $\xi$. Note that $\xi$ and $e^{i\alpha}\xi$ for $\alpha\in \mathbb{R}$ are mapped to the same $\xi\xi^\dagger$, but up to multiplication of $\xi$ by such a U(1) phase there is a one-to-one correspondence between the set of $\xi$ and the set of $\xi\xi^\dagger$.
Therefore we can view functions on the set of $\xi$ that are invariant under multiplication of $\xi$ by a U(1) phase, such as  $|I_1|,|I_2|,|I_{2A}|,|I_{2B}|$ and $|I_3|$, as functions on the set of matrices $\xi\xi^\dagger$. The matrices $\xi\xi^\dagger$ transform under the spinor representation $S(\Lambda_A)\otimes S(\Lambda_B)$ of local proper orthochronous Lorentz transformations $\Lambda_A\otimes \Lambda_B$ as $\xi\xi^\dagger\to S(\Lambda_A)\otimes S(\Lambda_B)\xi\xi^\dagger S(\Lambda_A)^\dagger\otimes S(\Lambda_B)^\dagger$.

Given the set of matrices $\xi\xi^\dagger$ we can consider conical sums of such matrices, i.e., $\sum_k p_k\xi_k\xi^\dagger_k$ where $p_k\geq 0$. These conical sums are positive semi-definite Hermitian matrices and are the analogues of density matrices. Such a density matrix $\rho$ can in general be decomposed as a conical sum of rank 1 matrices $\xi\xi^\dagger$ in more than one way, i.e., the decomposition is not unique.

Any continuous real valued function $g$ defined on the rank 1 matrices $\xi\xi^\dagger$ can be extended to a function on the set of density matrices as a {\it convex roof extension} \cite{lima,wakker,uhlmannn}. The convex roof extension $g^{\cup}$ of $g$ is defined as

\begin{eqnarray}
g^{\cup}(\rho)\equiv \inf_{p_k\geq 0,\xi_k|\rho=\sum_k p_k \xi_k\xi_k^\dagger}\sum_k p_k g(\xi_k\xi^\dagger_k),
\end{eqnarray} 
where the infimum is taken over all possible decompositions of $\rho$. While the formal definition of a convex roof extension is straightforward it may be difficult to calculate.

Now consider a spinor representation of local proper orthochronous Lorentz transformations acting on $\rho$ and the $\xi_k$ as $\rho\to\rho'$ and $\xi_k\to\xi'_k$. Since Lorentz transformations are invertible it follows that $\rho$ can be decomposed as $\rho=\sum_k p_k\xi_k\xi^\dagger_k$ if and only if $\rho'$ can be decomposed as $\rho'=\sum_k p_k\xi'_k\xi'^\dagger_k$. Therefore the set of decompositions of $\rho$ is in one-to-one correspondence with the set of decompositions of $\rho'$. 

Next assume that the function $g$ is invariant under the spinor representations of local proper orthochronous Lorentz transformations. Then $g(\xi_k\xi^\dagger_k)=g(\xi'_k\xi'^\dagger_k)$ for all $\xi_k$ and thus $\sum_k p_k g(\xi_k\xi^\dagger_k)=\sum_k p_k g(\xi'_k\xi'^\dagger_k)$ for all $p_k,\xi_k$. Therefore, it follows that the convex roof extension $g^{\cup}$ of $g$ is
invariant under the spinor representations of local proper orthochronous Lorentz transformations. 
Examples of such functions $g$ that can be given Lorentz invariant convex roof extensions are $|I_1|,|I_2|,|I_{2A}|,|I_{2B}|$ and $|I_3|$. Note that since  $|I_1|,|I_2|,|I_{2A}|,|I_{2B}|$ and $|I_3|$ are identically zero for product states their convex roof extensions are identically zero for any $\rho$ that can be decomposed as a conical sum $\sum_k p_k \xi_k\xi^\dagger_k$ where all $\xi_k\xi^\dagger_k$ correspond to product states, i.e., any separable state.

Density matrices can be used to describe states that are incoherent mixtures, for example states where there are uncertainties in the preparation procedure. If no uncertainty exists in the momentum degrees of freedom and all states in the incoherent mixture have the same definite momenta, any decomposition of the density matrix involves only $\xi_k\xi^\dagger_k$ with the same definite momenta. Therefore all such $\xi_k\xi^\dagger_k$ transform the same way when acted on by the evolution operator.
Thus in this case the convex roof extension $g^{\cup}$ of a Lorentz invariant $g$ is not only Lorentz invariant but also invariant under the same class of unitary evolutions as $g$. For incoherent mixtures without definite momenta the convex roof extensions of the Lorentz invariants $|I_1|,|I_2|,|I_{2A}|,|I_{2B}|$ or $|I_3|$ are not defined since the Lorentz invariants themselves have not been defined for this case.

One way that density matrices on the form $\rho=\sum_k p_k\xi_k\xi^\dagger_k$ can appear is as the result of taking the partial trace over the momentum degrees of freedom.
If a state $\psi_{AB}$ with multiple momentum components is considered 

\begin{eqnarray}
\psi_{AB}=\sum_{\bold{k_A},\bold{k_B}}\sum_{j_A,j_B}\psi_{j_A,j_B,\bold{k_A},\bold{k_B}}\phi_{j_A}e^{i\bold{k_A}\cdot\bold{x_A}}\otimes \phi_{j_B}e^{i\bold{k_B}\cdot\bold{x_B}},
\end{eqnarray}
we can treat the terms corresponding to the different pairs of momenta $\bold{k_A},\bold{k_B}$ independently. From these we can construct a Hermitian $16|\bold{k_A},\bold{k_B}|\times 16|\bold{k_A},\bold{k_B}|$ matrix $M$ where $|\bold{k_A},\bold{k_B}|$ is the number of different pairs of momenta and choose it such that it is composed by $16\times 16$ block matrices $M_{\bold{k_A},\bold{k_B};\bold{m_A},\bold{m_B}}$ that are indexed by the pairs of momenta $\bold{k_A},\bold{k_B}$ and $\bold{m_A},\bold{m_B}$ and defined by

\begin{eqnarray}
M_{\bold{k_A},\bold{k_B};\bold{m_A},\bold{m_B}}=&&\left(\sum_{j_A,j_B}\psi_{j_A,j_B,\bold{k_A},\bold{k_B}}(t)\phi_{j_A}e^{i\bold{k_A}\cdot\bold{x_A}}\otimes \phi_{j_B}e^{i\bold{k_B}\cdot\bold{x_B}}\right)\nonumber\\\times &&\left(\sum_{l_A,l_B}\psi_{l_A,l_B,\bold{m_A},\bold{m_B}}^*\phi^\dagger_{l_A}e^{-i\bold{m_A}\cdot\bold{x_A}}\otimes \phi^\dagger_{l_B}e^{-i\bold{m_B}\cdot\bold{x_B}}\right).\nonumber\\
\end{eqnarray}
If we then take the partial trace over the pairs of momenta the result is a $16\times 16$ matrix

\begin{eqnarray}
&&\sum_{\bold{k_A},\bold{k_B}}M_{\bold{k_A},\bold{k_B};\bold{k_A},\bold{k_B}}\nonumber\\
=&&\sum_{\bold{k_A},\bold{k_B}}\left(\sum_{j_A,j_B}\psi_{j_A,j_B,\bold{k_A},\bold{k_B}}\phi_{j_A}e^{i\bold{k_A}\cdot\bold{x_A}}\otimes \phi_{j_B}e^{i\bold{k_B}\cdot\bold{x_B}}\right)\nonumber\\&&\phantom{kk}\times\left(\sum_{l_A,l_B}\psi_{l_A,l_B,\bold{k_A},\bold{k_B}}^*\phi^\dagger_{l_A}e^{-i\bold{k_A}\cdot\bold{x_A}}\otimes \phi^\dagger_{l_B}e^{-i\bold{k_B}\cdot\bold{x_B}}\right)\nonumber\\
=&&\sum_{\bold{k_A},\bold{k_B}}\left(\sum_{j_A,j_B}\psi_{j_A,j_B,\bold{k_A},\bold{k_B}}\phi_{j_A}\otimes \phi_{j_B}\right)\left(\sum_{l_A,l_B}\psi_{l_A,l_B,\bold{k_A},\bold{k_B}}^*\phi^\dagger_{l_A}\otimes \phi^\dagger_{l_B}\right)\nonumber\\
\equiv &&\sum_{\bold{k_A},\bold{k_B}}\xi_{\bold{k_A},\bold{k_B}}\xi^\dagger_{\bold{k_A},\bold{k_B}},
\end{eqnarray}
which is a conical sum of positive rank 1 Hermitian matrices $\xi_{\bold{k_A},\bold{k_B}}\xi^\dagger_{\bold{k_A},\bold{k_B}}$ corresponding to the different pairs of momenta. It is the analogue of a reduced density matrix for the spinorial degrees of freedom.

The convex roof extensions of $|I_1|,|I_2|,|I_{2A}|,|I_{2B}|$ or $|I_3|$ are thus invariants of the spinor representations of the local proper orthochronous Lorentz groups on the set of reduced density matrices constructed as partial traces over the momentum degrees of freedom. 
Note however, that as described in Appendix \ref{dwalin} the evolution generated by local Dirac Hamiltonians is conditioned on the momentum. Thus the density matrices constructed as partial traces over momentum degrees of freedom do in general not have any well defined transformation properties under such local evolution. Consequently, the convex roof extensions of $|I_1|,|I_2|,|I_{2A}|,|I_{2B}|$ or $|I_3|$ are in general not invariant under local unitary evolution generated by Dirac Hamiltonians on the set of reduced density matrices constructed as partial traces over the momentum degrees of freedom.

\section{A comment on the infinite dimensional representations of the Lorentz group}\label{kili}

In this work the spinor representation of the Lorentz group has been considered in the context of relativistic quantum mechanics. Several previous works \cite{czachor,alsing,mano,caban,caban3,leon,tessier,geng,caban2,terno2,ahn,tera,tera2,terno,adami} however have considered an infinite dimensional unitary representations of the Lorentz group acting on a Hilbert space, often in the context of a relativistic Quantum Field Theory formalism. Therefore we comment on one of the crucial differences between these descriptions.

In a relativistic Quantum Field Theory formalism the solutions to the Dirac equation are re-imagined as operator valued fields acting on an infinite dimensional Hilbert space. The Hilbert space basis vectors $|s,\bold{k}\rangle$ are labelled by {\it rest frame spin} $s$ and three-momentum $\bold{k}$ (See e.g. Ref. \cite{weiny}) and transform under an infinite dimensional unitary representation of the Lorentz group. Let $F(\Lambda)$ be the representation of a Lorentz transformation $\Lambda$ acting on this Hilbert space.
 If the representation $F(L_k)$ of a pure Lorentz boost $L_k$ to momentum $\bold{k}$ acts on a rest frame basis vector $|s,\bold{0}\rangle$ the resulting vector is defined as $|s,\bold{k}\rangle\equiv F(L_k) |s,\bold{0}\rangle$. Here the label for the spin does not change, i.e., the spin degree of freedom is still described by its rest frame value, and is thus independent of any boost. Together, the momentum and the rest frame spin completely specify the spin of the particle, but on its own the rest frame spin is insufficient.

If a general Lorentz transformation $\Lambda$ is considered we can use the boost independence of the rest frame spin to re-express the action of the representation $F(\Lambda)$ on a basis vector as

\begin{eqnarray}
F(\Lambda)|s,\bold{k}\rangle=F(\Lambda)F( L_k)|s,\bold{0}\rangle=F(L_{\Lambda k})F(L_{\Lambda k}^{-1}\Lambda L_k)|s,\bold{0}\rangle,\nonumber\\
\end{eqnarray}
where $L_{\Lambda k}$ is the pure Lorentz boost from the rest frame to the momentum $\Lambda\bold{k}$ of $F(\Lambda)|s,\bold{k}\rangle$ and $L_{\Lambda k}^{-1}$ is the pure Lorentz boost that brings the vector back to the rest frame, i.e., the inverse Lorentz boost to $L_{\Lambda k}$. The sequence of Lorentz transformations $L_{\Lambda k}^{-1}\Lambda L_k$ changes the momentum away from $\bold{0}$ and then back to $\bold{0}$ again. It is therefore a pure rotation called a Wigner rotation \cite{wigner}. The representation $F(L_{\Lambda k}^{-1}\Lambda L_k)$ of this rotation acts on the rest frame spin so that

\begin{eqnarray}
F(\Lambda)|s,\bold{k}\rangle &&=F(L_{\Lambda k})\sum_{s'}C_{s,s'}(\Lambda,\bold{k})|s',\bold{0}\rangle\nonumber\\
 &&=\sum_{s'}C_{s,s'}(\Lambda,\bold{k})|s',\Lambda\bold{k}\rangle,
\end{eqnarray}
where $C_{s,s'}(\Lambda,\bold{k})$ are the matrix elements of the representation of the Wigner rotation. 
The Wigner rotation is a function of both the Lorentz transformation $\Lambda$ and the initial momentum $\bold{k}$.
In other words this infinite dimensional representation of the Lorentz group acts on the rest frame spin conditioned on the momentum. 

Since the specification of the spin requires both the rest frame spin and the momentum, it is not in general meaningful to construct a reduced density matrix for the rest frame spin by partially tracing over the momentum.
By discarding the momentum degrees of freedom we loose information needed for a physical interpretation. In particular, the resulting object does not contain the information necessary to specify how it transforms under Lorentz transformations. It does not have any well defined transformation properties under any representation of the Lorentz group. The notion of a reduced rest frame spin state of a particle that is independent of the momentum degrees of freedom is thus not physically meaningful except for the case where all basis vectors in the state expansion have the same momentum. Then the same Wigner rotation acts on all rest frame spin labels and the reduced density matrix can be given a physical interpretation. 
Moreover, the rank of the  rest frame spin reduced density matrix can change under Lorentz transformations. Consider for example a state $1/\sqrt{2}(|s,\bold{k}_1\rangle+|s,\bold{k}_2\rangle)$. If we naively write the state as $1/\sqrt{2}(|\bold{k}_1\rangle+|\bold{k}_2\rangle)\otimes|s\rangle$ and then construct the reduced density matrix for the rest frame spin by taking the partial trace over the momentum we obtain the rank one matrix $|s\rangle\langle s|$. A Lorentz transformation of the state  in general result in a new state for which the rest frame spin reduced density matrix is no longer rank one because of the momentum dependence of the Wigner rotations. 
These issues with the physical interpretation of reduced spin density matrices were pointed out in Ref. \cite{terno}.   

For the same reason the notion of entanglement between the rest frame spin degrees of freedom of two particles does not have a physically meaningful description independently of the momentum degrees of freedom
unless the particles have definite momenta. If one traces out the momentum degrees of freedom in a system of two particles without definite momenta the resulting reduced density matrix for the two rest frame spins has no physical interpretation and no transformation law under the local Lorentz groups. As described in Ref. \cite{adami} the rank of such a reduced density matrix can change under Lorentz transformations.  Moreover the mathematical counterpart of entanglement of the rest frame spins as quantified by the Wootters concurrence in general changes  under Lorentz transformations \cite{adami}.

In contrast this work considers the finite dimensional spinor representation of the Lorentz group that does not act conditioned on momentum. Therefore the analogue of reduced density matrices for spinors are well defined and transform under the spinor representation of the local Lorentz groups as explained in Appendix \ref{con}.

\end{document}